\pdfoutput=1
\documentclass[10pt, 5p, twocolumn]{elsarticle}
\usepackage{amsmath, amsthm, amssymb}
\usepackage{graphicx}
\usepackage{epsfig}
\usepackage{latexsym} 
\usepackage{epic}     
\usepackage{eepic}    
\usepackage{balance}
\usepackage{subfigure}
\usepackage{colortbl}
\usepackage{color,soul}
\usepackage{algpseudocode}

\newtheorem{theorem}{Theorem}[section]
\newtheorem{lemma}[theorem]{Lemma}
\newtheorem{proposition}[theorem]{Proposition}

\newtheorem{mydef}[theorem]{Definition}

\makeatletter
\renewcommand\paragraph{\@startsection{paragraph}{4}{\z@}%
    {-3.25ex\@plus -1ex \@minus -.2ex}%
    {1.5ex \@plus .2ex}%
    {\normalfont\normalsize\bfseries}}
    \makeatother

\begin{document}

\begin{frontmatter}

\title{\LARGE All Colors Shortest Path Problem}

\author{Yunus Can Bilge} 
\ead{can.bilge@std.ieu.edu.tr}

\author{Do\u{g}ukan \c{C}a\u{g}atay} 
\ead{dogukan.cagatay@std.ieu.edu.tr}

\author{Beg\"{u}m Gen\c{c}} 
\ead{begum.genc@std.ieu.edu.tr}

\author{Mecit Sar{\i}} 
\ead{mecit.sari@std.ieu.edu.tr}

\author{H\"useyin Akcan\corref{cor1}} 
\ead{huseyin.akcan@ieu.edu.tr}

\author{Cem Evrendilek}
\ead{cem.evrendilek@ieu.edu.tr}

\cortext[cor1]{Corresponding author. Tel: +90 232 488 8287, Fax: +90 232 488 8475}

\address{
    Izmir University of Economics 
    35330, Bal\c{c}ova, Izmir, Turkey 
}

\begin{abstract}
All Colors Shortest Path problem defined on an undirected graph aims at finding a
shortest, possibly non-simple, path where every color occurs at least once,
assuming that each vertex in the graph is associated with a color known in advance.
To the best of our knowledge, this paper is the first to define and investigate
this problem. Even though the problem is computationally similar to
generalized minimum spanning tree, and the generalized traveling salesman
problems, allowing for non-simple paths where a node may be visited multiple
times makes All Colors Shortest Path problem novel and computationally unique.
In this paper we prove that All Colors Shortest Path problem is NP-hard, and
does not lend itself to a constant factor approximation. We also propose several
heuristic solutions for this problem based on LP-relaxation, simulated annealing,
ant colony optimization, and genetic algorithm, and provide extensive simulations
for a comparative analysis of them. The heuristics presented are not the standard
implementations of the well known heuristic algorithms, but rather sophisticated
models tailored for the problem in hand. This fact is acknowledged by the very
promising results reported.

\end{abstract}


\end{frontmatter}

\section{Introduction}

Given an undirected edge weighted graph where each vertex has an apriori assigned
color, All Colors Shortest Path (\emph{ACSP}) problem is defined as a generic problem
in which the aim is to find a shortest possibly non-simple path that starts from a
designated vertex, and visits every color at least once.
As the same node might need to be visited multiple times, the path is not necessarily
simple. This makes \emph{ACSP} a novel and unique problem
that has never been studied before to the best of our knowledge. As the problem
is generic enough, it can be applied to a broad range of possible
areas including mobile sensor roaming, path planning, and item collection.

In this paper, we study \emph{ACSP} problem, prove that the problem is
NP-hard, and that a constant factor approximation algorithm cannot exist
unless $P=NP$. An ILP formulation is developed for \emph{ACSP}, and
elaborate heuristic solutions to this optimization problem are also provided.
These heuristics are based on LP-relaxation, simulated annealing, ant colony
optimization, and genetic algorithm. An experimental study is carried out to
compare them, and report the results.

The remainder of the paper is organized as follows. In Section~\ref{sec:related},
we discuss the related work, and position our paper with respect to the state of
the art. In Section~\ref{sec:define}, we formally define the problem, and provide
the intractability proof along with an inapproximability result. Section~\ref{sec:ILP}
presents an ILP formulation for \emph{ACSP}. In Section~\ref{sec:heuristics},
we discuss the heuristic solutions we propose.
The experimental results are presented in
Section~\ref{sec:experiments}, and the paper is concluded in
Section~\ref{sec:conclusion}. 

\section{Related Work}\label{sec:related}

\emph{ACSP}, defined and investigated in this paper, has actually features that
make it look similar to a variety of problems studied extensively in the
literature, each of which, however, has one or more discrepancies making \emph{ACSP}
computationally unique. Among these, Generalized Minimum Spanning Tree
(\emph{GMST}) problem introduced in \cite{MLT95} is probably the most similar to
\emph{ACSP}. Given an undirected graph partitioned into a number of disjoint clusters,
\emph{GMST} problem is defined to be the problem of finding the minimum cost spanning tree
with exactly one node from every cluster. This problem has been shown to be
NP-hard in \cite{MLT95}, and some inaproximability results are presented in
\cite{P04}. Integer Linear Programming (\emph{ILP}) formulations for this problem
are presented in \cite{FLL02}, \cite{PSK05}, and \cite{PKS06}.
There exist formulations for also a variant of \emph{GMST} in \cite{DHC00}
and \cite{IRW99} where \emph{at least one} instead of \emph{exactly one}
node from each cluster is visited. We refer to the latter version as \emph{$\ell$-GMST}.
Even though there are such formulations, \emph{ACSP} still differs in the shape
of the solution. While \emph{ACSP} outputs a possibly non-simple path,
\emph{$\ell$-GMST} returns a tree.
Moreover, it can be easily noted that a minimum spanning tree returned
by \emph{$\ell$-GMST} can only give a rough estimate for the size of a
possibly non-simple shortest path visiting all the colors even when
\emph{ACSP} is required to return to the base it starts off as shown
in Figure~\ref{ACSPvsl-GMST}. When the nodes with the same color
are perceived as disjoint clusters so as to interpret this figure
as an instance of \emph{$\ell$-GMST}, the tree spanning nodes $1$ through
$6$ is the optimal solution to it with cost $5$. 

\begin{figure} [h!]
\centering
\includegraphics[width=\columnwidth] {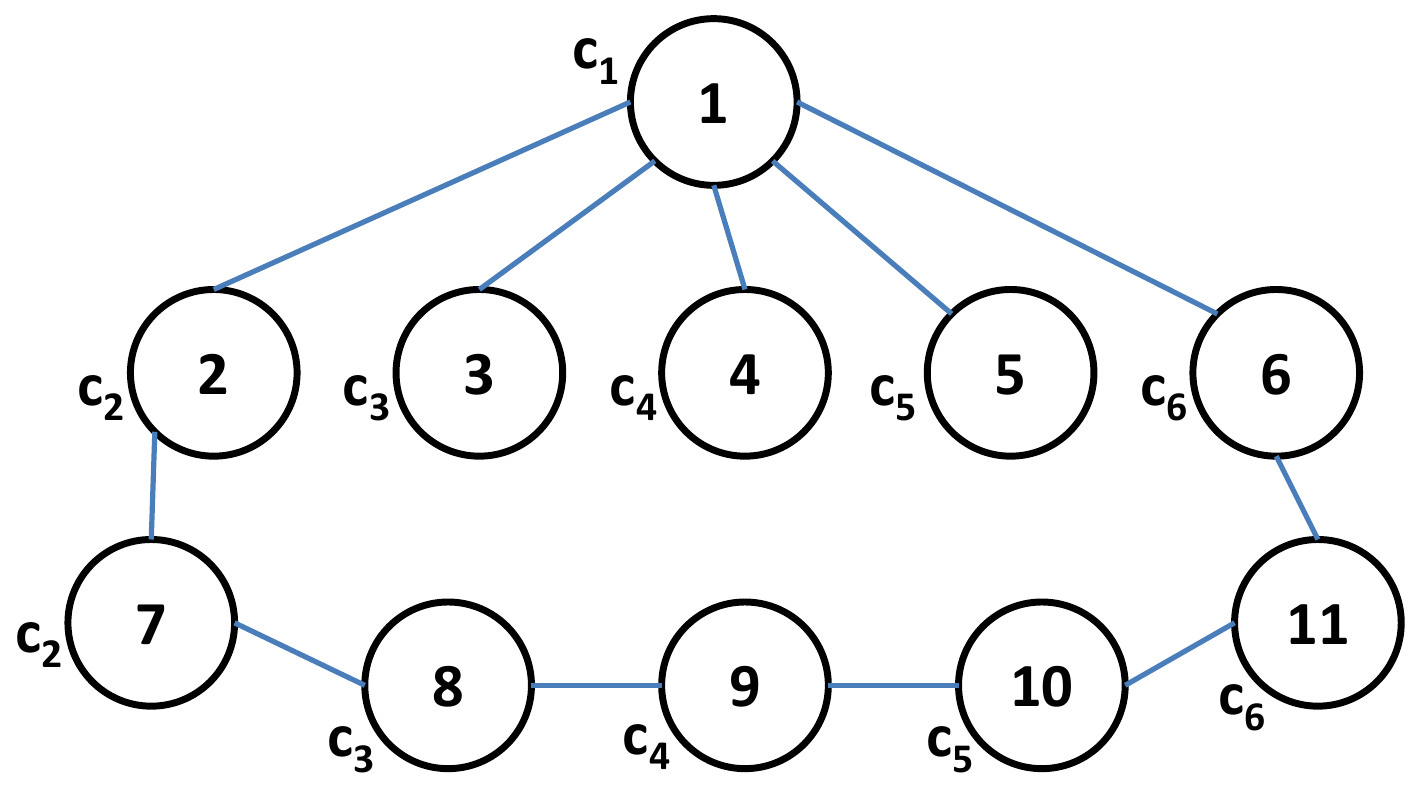}
\caption{An example graph corresponding to an instance of \emph{ACSP}. All the edges have
a weight of $1$, and the colors assigned to the nodes are shown next to them. Node $1$ is designated
as the base. The shortest path for this instance of \emph{ACSP} is $1, 2, 7, 8, 9, 10, 11$ which
has a length of $6$. When the path is constrained to return
to the base, however, the path length of the solution becomes 8.}
\label{ACSPvsl-GMST}
\end{figure}

Another problem seemingly similar to \emph{ACSP} is Generalized Traveling
Salesman Problem \emph{(GTSP)} formulated first in \cite{L69}. Given a group of
possibly intersecting clusters of nodes, \emph{GTSP} tries to find a shortest
Hamiltonian tour with at least one (or exactly one) visit to a node from every
cluster. An integer linear programming formulation for \emph{GTSP} when the distance
matrix is asymmetrical is given in \cite{LMN87}. In \cite{LMW93}, it is shown
that a given instance of \emph{GTSP} can be transformed into an instance of standard
\emph{TSP}. In \cite{FGT97}, \emph{GTSP} is noted to be NP-hard as standard
\emph{TSP} is a
specialization of \emph{GTSP} with clusters in the form of singleton nodes. It is also
surprising to note as \cite{BM02} demonstrates that \emph{GTSP} can be transformed
into standard \emph{TSP} very efficiently with the same number of nodes, but with a modified
distance matrix. \emph{ACSP} differs from also these variants of \emph{GTSP}, in that,
the nodes may be visited multiple times, and the path returned need not be a cycle.

\section{The Problem Definition}\label{sec:define}

\emph{ACSP} is modeled as a graph problem. The input to the problem is an undirected
edge weighted graph where each vertex is assigned a color known in advance. The goal
is then to find the shortest possibly non-simple path that visits every distinct color at least
once in this graph. The formal definition of the problem is given as:

\begin{mydef}
Given an undirected graph $G(V,E)$ with a color drawn from a set $C$ of colors assigned to
each node, and a non-negative weight associated with each edge, \emph{ACSP} is the problem
of finding the shortest (possibly non-simple) path starting from a designated base
node $s \in V$ such that every color occurs at least once on the path.
\end{mydef}

The weights $w_{i, j}$ where $(i, j) \in E$ in $G$ correspond to distances. We will use
the words \emph{weight}, \emph{cost}, and \emph{distance} interchangeably
throughout the paper. The cost of a solution to an instance of \emph{ACSP} is
simply the length of the path returned.

\emph{ACSP} can easily be shown to be NP-hard by a trivial polynomial time
reduction from \emph{Hamiltonian Path} (HP) problem which is well-known to be
NP-complete \cite{GJ79}. Given an undirected graph $G(V, E)$, \emph{HP} is
defined to be the problem of deciding whether it has a Hamiltonian path, namely,
a simple path that visits every node in the graph exactly once.

\subsection{NP-hardness of ACSP}

Given an instance of \emph{HP}, it can be transformed to the corresponding
instance of \emph{ACSP} as follows: Let the graph in the given \emph{HP} instance
be denoted by $G(V, E)$. A new graph $G'(V \cup \{s\},E \cup
\{(s, v) | v \in V \})$ is obtained by adding to $G$ a new node $s$, and also the
edges from $s$ to all the original nodes in $G$. Next, a distinct color from $C=\{c_1,
c_2, ..., c_{|V|+1} \}$ is assigned to each and every node in $G'$. The weights
associated with all the edges in $G'$ are finally set to one. We can now state
the following lemma.

\begin{lemma}\label{lemma:HPiffACSP}
A given instance of \emph{HP} represented with $G(V,E)$ has a solution if and
only if the corresponding instance of \emph{ACSP} obtained through the lines of
transformation just depicted has a solution with length $|V|$.
\end{lemma}

\begin{proof}
Let us first prove the only if part. When the given instance of \emph{HP} has a
solution, there must exist a Hamiltonian path $P$ in $G$ given by
$v_{\pi(1)}v_{\pi(2)} ... v_{\pi(i)}v_{\pi(i+1)} ... v_{\pi(|V|)}$ of length
$|V|-1$. As $P$ is a Hamiltonian path, the permutation $\pi$ of nodes in $V$ is
such that the edges $( v_{\pi(i)}, v_{\pi(i+1)}) \in E$ for all $ i \in
\{1..|V|-1\}$. If we let $C$, and $G'(V', E')$ denote the set of $|V|+1$ colors,
and the transformed graph respectively in the corresponding instance of
\emph{ACSP}, it is then possible to construct the path $P'=sP$ in $G'$ with
total path length $|V|$ where $s \in V$ is designated as the base node.
This is apparently the shortest path visiting all distinct colors at least once.

In order to prove the if part, let us assume that we have a shortest path
of length $|V|$ that starts with node $s$ in the corresponding instance of
\emph{ACSP}. Since the total number of colors that needs to be visited is
$|V|+1$, each distinct color, and hence, the corresponding node occurs
exactly once on this path. The removal of node $s$ readily specifies a
Hamiltonian path in $G$ of the given \emph{HP} instance.
\end{proof}

The following theorem can hence be stated now.

\begin{theorem}\label{theorem:np-hard}
\emph{ACSP} is NP-hard.
\end{theorem}

\begin{proof}
It is a direct consequence of Lemma~\ref{lemma:HPiffACSP}.
\end{proof}

Having learned about the NP-hardness of \emph{ACSP}, a possible next step
is to explore its approximability. With this objective in mind, our attention was
drawn to \emph{$\ell$-GMST} problem having a similar computational structure.
While \emph{$\ell$-GMST} looks for the minimum cost spanning tree, \emph{ACSP}
seeks out a possibly non-simple path with at least one node from every cluster
provided that the nodes with the same color are interpreted as disjoint clusters.
The following observation is first made to associate the optimal values of the
respective solutions attained by both problems when fed with the same input.
It is then used to report a result regarding the approximability of \emph{ACSP}.

\begin{proposition}\label{prop:sandwich}
Let $I$ correspond to an input identified by an undirected edge weighted
graph $G(V, E)$, and a function $\kappa:V \to \{1, 2, ..., k\}$ mapping the
vertices to colors.
For $1 \le i \le k$, $V_i=\{v \in V|\kappa(v)=i\}$ induce clearly a set of
$k$ disjoint clusters, which in turn allows for a proper interpretation of $I$
by $\ell$-GMST. Then,
\[
\text{opt}_{\ell\text{-GMST }}(I) \le
\min_{j \in V} \{ \text{opt}_{\text{ACSP }}(I_j) \} <
2 * \text{opt}_{\ell\text{-GMST }}(I)
\]
holds for all valid instances $I$, where $I_j$ is obtained from $I$ by designating
$j \in V$ as the base node, and $\text{opt}_{A}$ returns the cost of the optimal
solution to its argument interpreted as
an instance of either one of the two problems as dictated by the subscript $A$.
\end{proposition}

\begin{proof}
Let us assume that the first inequality in the proposition does not hold,
and, there is an instance $I$ for which
$\textit{opt}_{\ell\textit{-GMST }}(I) >
\min_{j \in V} \{ \textit{opt}_{\textit{ACSP }}(I_j) \}$.
Let us suppose that $s \in V$ is a node that minimizes the right-hand
side of this inequality. In that case, the solution to \emph{ACSP} with cost
$\textit{opt}_{\textit{ACSP }}(I_s)$ can be easily reworked,
by simply eliminating any cycles, and duplicate edges, into a tree $T'$.
$T'$ is clearly a solution to \emph{$\ell$-GMST} for the instance $I$
with cost less than $\textit{opt}_{\ell\textit{-GMST }}(I)$, and hence,
contradicting the assumption.

Let us assume now the latter inequality does not hold.
This, for at least one instance of input $I$, leaves us with
$\min_{j \in V} \{ \textit{opt}_{\textit{ACSP }}(I_j) \} \ge
2 * \textit{opt}_{\ell\textit{-GMST }}(I)$. Let us also
assume that the tree $T'(V', E')$ is a solution to \emph{$\ell$-GMST}
with cost $\textit{opt}_{\ell\textit{-GMST }}(I)$ for instance $I$.
Rooting $T'$ at some $s \in V'$, a possibly non-simple path starting
from $s$ could be constructed visiting all the nodes in it by a depth-first search.
This path, however, forms a solution to \emph{ACSP} for instance $I_s$
with cost strictly less than $2 * \textit{opt}_{\ell\textit{-GMST }}(I)$
as no edge gets visited more than twice,
and there exists at least one edge that is visited exactly once given that
the return to the base is not performed upon hitting the last leaf node in $V'$.
This, however, contradicts the assumption.
\end{proof}

It is shown in \cite{IRW99} that \emph{$\ell$-GMST}, referred to as CLASS TREE problem
in the paper, does not have a constant-factor polynomial time approximation algorithm (\emph{apx})
unless $P=NP$.

\begin{theorem}\label{theorem:inapproximable}
\emph{ACSP} does not have a constant-factor polynomial time approximation algorithm
unless $P=NP$.
\end{theorem}

\begin{proof}
Let us assume, to the contrary, that \emph{ACSP} has an apx denoted
by $\textit{apx}_{\textit{ACSP}}$. Based on this assumption, an apx
for \emph{$\ell$-GMST} can be shown to also exist, and hence a
contradiction, as follows.

Given any valid input $I$ for \emph{$\ell$-GMST}, consisting of an
undirected graph $G(V, E)$ along with disjoint clusters $V_i \subseteq V$
with $1 \le i \le k$, we denote by $I_j$ the input
for \emph{ACSP} obtained from $I$ by designating $j \in V$ as the base.
The initial assumption with regard to the existence of an apx suggests by definition
\[
\textit{opt}_{\textit{ACSP }}(I_j) \le
\textit{apx}_{\textit{ACSP }}(I_j) \le
c*\textit{opt}_{\textit{ACSP }}(I_j)
\]
for some constant $c>1$, and all valid input $I_j$ where $j \in V$. Taking
the minimum over all $j \in V$, we obtain

{\footnotesize
\[
\min_{j \in V} \{ \textit{opt}_{\textit{ACSP}}(I_j) \} \le
\min_{j \in V} \{ \textit{apx}_{\textit{ACSP}}(I_j) \} \le
c*\min_{j \in V} \{ \textit{opt}_{\textit{ACSP}}(I_j) \}.
\]
}

\noindent Combining this result with Proposition~\ref{prop:sandwich},
\[
\textit{opt}_{\ell\textit{-GMST }}(I) \le
\min_{j \in V} \{ \textit{apx}_{\textit{ACSP }}(I_j) \} <
2*c*\textit{opt}_{\ell\textit{-GMST }}(I)
\]
is readily obtained. It should be noted that the minimization
over $\textit{apx}_{\textit{ACSP }}(I_j)$ involves running
the constant-factor approximation for \emph{ACSP} separately
for each $j \in V$, and the total time, even though amplified by
a factor of $|V|$, is still polynomial in the size of a given instance.
Therefore, the last inequality implies, by definition, a $2c$-factor
apx for \emph{$\ell$-GMST}.
This, however, is a contradiction, and hence, the proof.
\end{proof}

\section{ILP Formulation of \emph{ACSP}}\label{sec:ILP}
In this section, an Integer Linear Programming formulation of \emph{ACSP}
is presented. To this end, we start by making the following observation first.

\begin{proposition}\label{proposition:maxOneVisit}
In an optimal solution to any instance of \emph{ACSP}, no edge can be visited
more than once in any given direction.
\end{proposition}

\begin{proof}
We assume that $p$ is a possibly non-simple path with the shortest distance,
forming a solution to a given instance of \emph{ACSP}. Contrary to the proposition,
we proceed by assuming that an edge $(i, j)$ is traversed more than once in the
direction from node $i$ to node $j$. Highlighting the first two occurrences of this
edge, then, the path can be represented as $p=s,x,i,j,y,i,j,z$ where $s$ is the
base, and $x$, $y$, and $z$ are sequences of zero or more nodes with edges
in between consecutive nodes. It should be noted that neither $x$ nor $y$
are allowed, by the assumption, to have
any occurrences of $i$, and $j$ consecutively in this order. We can, in that case,
construct a new path $p'=s,x,i,y^R,j, z$ with $y^R$ corresponding, in reverse
order, to the sequence of nodes in $y$. This new path, $p'$, visiting the same
set of nodes as $p$, however, is shorter by $2*w_{i,j}$ than $p$. This contradicts
the optimality of $p$, and hence, proving the proposition.
\end{proof}

Proposition~\ref{proposition:maxOneVisit} allows for an ILP formulation to
\emph{ACSP} where tracking down whether an edge is visited
as part of an optimal solution
in either one of the two possible directions
becomes possible by employing a binary decision variable.
This observation, coupled with the motivation to come up
with a compact ILP model, form the basis of the transformation
to be described next. At the heart of the transformation is
the replacement of each undirected edge in a given instance
of \emph{ACSP} with two directed edges, and hence the adoption
of a directed graph view as a substitute in the ILP formulation.

Let us assume that an instance of
\emph{ACSP} is given, as determined by an undirected edge weighted graph $G(V, E)$,
the designated base vertex $s \in V$, and $\kappa:V \to C$ mapping the vertices
$V=\{1, ..., n\}$ to colors $C=\{1, ..., k\}$. Finally, the weights associated with the edges
in $G$ are denoted by $w_{i, j}$ for all unordered pairs $(i, j) \text{ (or } \{i, j\} \text{)} \in E$.
It is therefore implicitly assumed that $w_{i, j}=w_{j, i}$ for all $(i, j) \in E$.

In transforming $G(V, E)$ to a directed graph $G'(V', E')$ to be used in the ILP formulation,
we first introduce two new nodes numbered $0$ as the source, and $n+1$ as the sink,
setting effectively $V'= V \cup \{0, n+1\}$ in $G'$. Besides, the source, and the sink are both
assigned to a new color $0$, extending the color set to $C' = C \cup \{0\}$. With the addition
of the new color, $\kappa$ is also augmented accordingly with
$\kappa(0)=\kappa(n+1)=0$. Then, a directed edge $(0, s)$ from the new source
to the base $s$, as well as directed edges $(i, n+1)$ to the sink, for all $i \in V$ in $G$, are added
into $G'$ with their weights set to $0$. Lastly, each undirected edge $(i, j) \in E$ is replaced by two
directed edges $(i, j)$ and $(j, i)$ in $G'$ with both of whose weights initialized to the weight of
the original edge. With this final step, the transformation sets
$E' = \{(i, j), (j, i) | (i, j) \in E\} \cup \{(0, s)\} \cup \{(i, n+1) | i \in V\}$ in $G'$.
Continuing to use the same notation for weights in $G'$, $w_{0, s}=0$,
and $w_{i, n+1}=0$ for all $i \in V$
are added after the existing $w_{i, j}=w_{j, i}$ for all unordered pairs $(i, j) \in E$.

Any possibly non-simple path, $p$,
starting from the designated base $s$, and visiting all colors at least once in $G$, corresponds
precisely to the path $0, p, n+1$ in $G'$, where nodes $0$, and $n+1$ are the source, and the
sink respectively. In the same way, a possibly non-simple path $p=0, p', n+1$ in $G'$, where
$p'$ is a possibly non-simple path starting at $s$, and with length at least one, corresponds
to $p'$ in $G$. As a result, the feasible solutions in $G$, and $G'$ will be in one-to-one correspondence,
as long as the ILP formulation of \emph{ACSP} can place a restriction on any feasible solution
in $G'$ to start from the source, and to terminate at the sink. Moreover, these corresponding
solutions have both the same cost.
It is hence obvious that a solution to an instance of \emph{ACSP} on $G$ as given above
is optimal if and only if the corresponding solution on the transformed instance employing $G'$
is also optimal. 

The ILP formulation for a given instance of \emph{ACSP} can now be stated with reference
to the transformation described above.
\begin{align}
&\text{minimize} \sum_{(i, j) : (i, j) \in E'} x_{i, j}*w_{i, j}
\label{ILP:objective}
\\
&\text{subject to} \nonumber
\\
& x_{0, s} = 1
\label{ILP:sourceToBase}
\\
& \sum_{\substack{(i, j) : (i, j) \in E' \\ \wedge~\kappa(j)=c}} x_{i, j} \ge 1 && ,\; \forall c \in C'
\label{ILP:visitEveryColor}
\\
& \sum_{\substack{j : (j, i) \in E' \\ \hphantom{(i, j) : (i, j) \in E'}}} x_{j, i}
= \sum_{j : (i, j) \in E'} x_{i, j} && ,\; \forall i \in V
\label{ILP:inEqualOut}
\\
& y_j \ge x_{i, j} && ,\; \forall (i, j) \in E'
\label{ILP:nodeVisitedIfEdgeEnters}
\\
&  \sum_{\substack{i : (i, j) \in E' \\ \hphantom{(i, j) : (i, j) \in E'}}} x_{i, j} \ge y_j && ,\; \forall j \in V' \setminus \{0\}
\label{ILP:edgeEntersIfNodeVisited}
\\
& \sum_{\substack{j : (j, i) \in E' \\ \hphantom{(i, j) : (i, j) \in E'}}} f_{j, i}
=y_i + \sum_{j : (i, j) \in E'} f_{i, j} && ,\; \forall i \in V
\label{ILP:flow}
\\
& x_{i, j} \le f_{i, j} \le (n+1)*x_{i, j} && ,\; \forall (i, j) \in E'
\label{ILP:boundFlow}
\\
& x_{i, j} \in \{0, 1\} && ,\; \forall (i, j) \in E'
\label{ILP:x}
\\
& y_i \in \{0, 1\} && ,\; \forall i \in V'\setminus \{0\}
\label{ILP:y}
\\
& f_{i, j} \in \{0, 1, ..., n+1\} && ,\; \forall (i, j) \in E'
\label{ILP:f}
\end{align}
The objective in this formulation is to minimize the sum of the weights
over all the directed edges that have been visited as shown in \eqref{ILP:objective}.
The binary variable $x_{i, j}$ is set to $1$
when the directed edge $(i, j)$ is visited, and to $0$ otherwise.
It should be noted that all the edges involving the source, and the sink,
introduced later in the transformation, with weight zero have no effect on the objective.
Constraint~\eqref{ILP:sourceToBase} ensures that the edge from the source
to the base is always a part of any feasible solution.
Therefore, any feasible path always starts from the source,
and then moves straight to the base.
Constraint~\eqref{ILP:visitEveryColor} demands for each
distinct color that the number of the visited edges directed at the
nodes with this same color is at least one.
As a result, every distinct color gets visited at least once.
As the constraint must also hold for color $0$, any feasible
path is guaranteed to terminate at the sink.
Constraint~\eqref{ILP:inEqualOut} is used to make sure that the
number of the visited edges that enter into any node $i$ in $G$ is equal to the
number of the visited edges that leave it. This obviously holds for all the nodes, but
the source, and the sink in $G'$. The main ingredient in enforcing the
shape of the solution to a possibly non-simple path is this constraint.
Constraints~\eqref{ILP:nodeVisitedIfEdgeEnters}, and \eqref{ILP:edgeEntersIfNodeVisited}
establish collectively the rules associated with the variables $y_j$
for all $j \in V' \setminus \{0\}$. The binary decision variable $y_j$ is set to $1$
if and only if node $j$ has been visited in a feasible solution.
Constraint~\eqref{ILP:nodeVisitedIfEdgeEnters} simply asserts
that visiting an edge $(i, j)$ is an implication of visiting node $j$
while Constraint~\eqref{ILP:edgeEntersIfNodeVisited}
predicates the converse.
Constraint~\eqref{ILP:flow}, along with \eqref{ILP:boundFlow}, is
used to eliminate any possible sub-tours, and to ensure connectedness
to the base. Constraint~\eqref{ILP:flow} employs non-negative integer valued
flow variables, denoted by $f_{i, j}$ for all edges $(i, j) \in E'$.
It enforces the total flow into a visited node to be equal to
one greater than the total flow out of that node.
In formulating this constraint, it is assumed that
the source supplies a limited amount of flow to distribute
to those nodes that are visited in any feasible solution. Hence,
each node visited consumes a unit flow.
Constraint~\eqref{ILP:boundFlow} is in charge of regulating
the flow values. A flow is associated with an edge if and only if
that edge is part of a solution. As the flow is conserved at
all the nodes in the original graph, the base node $s$ is no
exception. Coupled with the fact that each node visited consumes
a unit flow, the edge $(0,s)$ should carry
as many unit flows as there are nodes to visit. Excluding the source
leaves us with a maximum of $n+1$ nodes, and hence, the factor
in \eqref{ILP:boundFlow}. Finally, the constraints \eqref{ILP:x},
\eqref{ILP:y}, and \eqref{ILP:f} are the integrality constraints
for the decision variables $x_{i, j}$, $y_i$, and $f_{i, j}$ respectively.

\section{Heuristic solutions}\label{sec:heuristics}
In this section, we describe our heuristic solutions to the intractable
\emph{ACSP} problem. Section~\ref{sec:LP-relaxation} explains several
heuristic solutions based on LP-relaxation.
Simulated annealing, ant colony optimization, and genetic algorithm based
heuristic solutions to \emph{ACSP} are presented later in sections
\ref{sec:sim-ann}, \ref{sec:ant-col}, and \ref{sec:genetic} respectively.

\subsection{LP Relaxation}\label{sec:LP-relaxation}

The given ILP formulation, \eqref{ILP:objective} through \eqref{ILP:f},
is relaxed to an LP by replacing the integrality constraints
\eqref{ILP:x}, \eqref{ILP:y}, and \eqref{ILP:f} with
\begin{align}
& 0 \le x_{i, j} \le 1 && ,\; \forall (i, j) \in E'
\tag{\ref{ILP:x}$'$}
\label{ILP:x'}
\\
& 0 \le y_i \le 1 && ,\; \forall i \in V'\setminus \{0\}
\tag{\ref{ILP:y}$'$}
\label{ILP:y'}
\\
& 0 \le f_{i, j} \le n+1 && ,\; \forall (i, j) \in E'
\tag{\ref{ILP:f}$'$}
\label{ILP:f'}
\end{align}
respectively.
Now, the decision variables can take on real values.

We propose several heuristics based on rounding the solutions
to this LP relaxation. Having learned from Theorem~\ref{theorem:inapproximable}
that a constant-factor approximation does not exist for \emph{ACSP},
we explore strategies based on iterative rounding, rather than
typical one-shot rounding.

The first heuristic, called \emph{$\text{LP}_{\text{x}}\text{ACSP}$},
after obtaining the optimal solution to the LP relaxation, finds the
maximum value strictly less than one among all $x_{i, j} \in E'$.
The set of all indexes for which this maximum is attained
is denoted by $\mu$
(i.e., $\mu=\underset{((i,j) \in E') \wedge (0<x_{i, j}<1)}{\operatorname{arg\,max}} x_{i, j}$).
Next, the LP relaxation formulation at hand is augmented
with the additional constraints in the form of $x_{i, j}=1$ for all $(i, j) \in \mu$.
Finally, a subsequent call to LP for the extended formulation is issued.
Hence, the job, in this subsequent call, becomes finding the shortest possibly
non-simple path that fulfills not only the previous set of constraints but
also passes through every additional edge explicitly dictated by
the added constraints. This process is repeated until no fractional
values to process are left, and hence $\mu=\emptyset$.

The other two heuristics, called \emph{$\text{LP}_{\text{f}}\text{ACSP}$},
and \emph{$\text{LP}_{\text{f/x}}\text{ACSP}$} use exactly the same
strategy described above except for how $\mu$ is computed before a call
to LP. While \emph{$\text{LP}_{\text{f}}\text{ACSP}$}
relies on the flow variables
($\mu=\underset{((i,j) \in E') \wedge (0<x_{i, j}<1)}{\operatorname{arg\,max}} f_{i, j}$)
in deciding which additional $x_{i, j}$ values to round before the next iteration,
\emph{$\text{LP}_{\text{f/x}}\text{ACSP}$} bases its decision on the ratios
of $f_{i, j}/x_{i, j}$
($\mu=\underset{((i,j) \in E') \wedge (0<x_{i, j}<1)}{\operatorname{arg\,max}} \frac{f_{i, j}}{x_{i, j}}$).

\subsection{Simulated Annealing}\label{sec:sim-ann}

We develop another heuristic solution for \emph{ACSP},
based on Simulated Annealing (SA) \cite{Kirkpatrick1983}. This new
heuristic, called \emph{SA-ACSP}, can be described in three primary parts:
%

\textbf{(1) Choosing an initial random path:} The general outline of the algorithm
for \emph{SA-ACSP} is given in Figure~\ref{fig:sim-ann-algorithm}.
The algorithm starts with a random possibly non-simple path that visits every color
at least once. Such a path is constructed by randomly extending an existing
path, originating from the base, until it visits every color at least once.
This process is performed only once, in line~\ref{line:randomPath},
at the start of each iteration
in the while loop in lines~\ref{line:startWhile} through \ref{line:endWhile}.

\textbf{(2) Generating neighbors:} 
We generate a neighbor by removing the last node in the current state, and
then, adding to a random position in the path the closest node with the same
color as the removed node.

\textbf{(3) Selecting the best path:} 
Starting from an initial temperature, denoted by $T$ in the algorithm in
Figure~\ref{fig:sim-ann-algorithm}, the system is cooled down until a
frozen state is reached, where the temperature is close to zero.
Cooling down is done by decreasing the temperature slightly at each
iteration as seen in line~\ref{line:coolDown}. The symbol $R$ there
corresponds to the cooling rate.
The energy of each state is defined by the cost of the path selected at
that stage. As the temperature approaches to that of a frozen state,
\emph{SA-ACSP} keeps exploring the neighbors. At each iteration,
a neighbor solution is discovered in the search space, and chosen probabilistically
according to Metropolis Criterion \cite{Metropolis},
\begin{equation*}
p(\Delta E) = e^{-\Delta E / kT},
\end{equation*}
where $k$ is Boltzmann's constant, $T$ is the temperature, and $\Delta E$ is the
difference between the energies of the current, and the neighbor solutions.
If the total energy decreases, the new state is assumed right away. Otherwise,
the system chooses to go to the new state according to the probability
that is produced by Metropolis criterion. 

\begin{figure}[h!]
\begin{algorithmic}[1]
\begin{small}
\Procedure {Sa-Acsp}{}
	\State $\Delta E \leftarrow 0;$
	\State $\textit{bestCost} \leftarrow \infty;$
	\State $T \leftarrow \textit{SetInitialTemperature}();$
	\State $\textit{iterationCount} \leftarrow \textit{noOfNodes}*\textit{noOfColors}/5;$
	\While { $T \ge \textit{freezingTemp}$} \label{line:startWhile}
		\State $\textit{localBestPath} \leftarrow \textit{findARandomPath}();$\label{line:randomPath}
		\State $\textit{localBestCost} \leftarrow \textit{costOf}(\textit{localBestPath});$
		\For {$i \gets 0 \textbf{ to } \textit{iterationCount}$}
			\State $\textit{nextPath} \leftarrow \textit{findANeighbourPath}();$
			\State $\textit{nextCost} \leftarrow \textit{costOf}(\textit{nextPath});$
			\State $\Delta E \leftarrow \textit{nextCost} - \textit{localBestCost};$
			\If {$\Delta E < 0$}
				\State $\textit{localBestPath} \leftarrow \textit{nextPath};$
				\State $\textit{localBestCost} \leftarrow \textit{nextCost};$
			\Else
				\State $r \leftarrow \textit{Random}(0,1);$
				\If{$r < e^{-\Delta E/T}$}
					\State $\textit{localBestPath} \leftarrow \textit{nextPath};$
					\State $\textit{localBestCost} \leftarrow \textit{nextCost};$
				\EndIf
			\EndIf
		\EndFor
		\If{$\textit{localBestCost} < \textit{bestCost}$}
			\State $\textit{bestPath} \leftarrow \textit{localBestPath};$
			\State $\textit{bestCost} \leftarrow \textit{localBestCost};$
		\EndIf
		\State $T \leftarrow T * R;$\label{line:coolDown}
	\EndWhile \label{line:endWhile}
	\State \textbf{return} \textit{bestCost};
\EndProcedure
\end{small}
\end{algorithmic}
\caption{\emph{SA-ACSP}: Heuristic based on simulated annealing.}
\label{fig:sim-ann-algorithm}
\end{figure}

\subsection{Ant Colony Optimization}\label{sec:ant-col}
In this section, we present the details of how Ant Colony Optimization (ACO)
\cite{Dorigo1997,Dorigo1996} is applied to obtain \emph{ACO-ACSP},
another heuristic solution, to \emph{ACSP} problem.

In \emph{ACSP}, each color needs to be visited at least once.
Therefore, the ant colony optimization algorithm is implemented
to visit multiple food types, where each food type corresponds
to a distinct color. In other words, when an ant leaves the nest,
its search is not over until it finds a path that passes over every
food type. The base node is chosen as the nest of all the ants,
and at each iteration, the entire colony of ants is released from
this nest to the graph.

The random movements of the ants, while visiting a node, are governed
by an edge selection procedure. Depending on whether there is any
trace of pheromone on an incident edge, the ants compute two types
of edge selection probabilities. In the first one, when
there is no pheromone on any incident edge, the ants make the selection
based on the edge costs (or distances) using the following formula,
\begin{equation*}
prob_{i,j}=\frac{c_0 - w_{i, j}}{\sum_{k:(i, k) \in E}(c_0 - w_{i, k})},
\end{equation*}
where $prob_{i, j}$ is the selection probability of edge $(i, j) \in E$,
and $c_0$ is a constant.
The second case occurs when the pheromone level on at least one incident
edge is not zero. In this case, the edge selection probability calculation
is performed, based on the pheromone levels, as:
\begin{equation*}
prob_{i, j}=\frac{(D_{i, j})^{\beta}*\sum_{k \in C} (Ph_{i, j}(k))^{\alpha}}
{\sum_{q:(i, q) \in E} \big[ (D_{i, q})^{\beta}*\sum_{k \in C} (Ph_{i, q}(k))^{\alpha}\big]},
\end{equation*}
where $Ph_{i, j}(k)$ is the pheromone level on edge $(i, j) \in E$
associated with color $k \in C$, $\alpha$, and $\beta$ are user defined parameters
with $0 \le \alpha \le \beta \le 1$, and the desirability $D_{i, j}$ of edge
$(i, j) \in E$ is defined to be inversely proportional to the edge's cost
as $D_{i, j}=1/w_{i, j}$.

Pheromone level updates are carried out in two different
ways, namely, the local, and the global updates.
The local updates are applied to all the edges selected
because each ant secretes pheromone as it moves on the edges.
Moreover, the pheromones are not stored only on the edges.
Ants also have some pheromones within themselves, and
their levels drop while they are secreted by the ants during
their traversal. Therefore, the local updates are performed
on the edges as well as the ants selecting them. While the
local update to the pheromone level corresponding to color
$k \in C$, after the selection of edge $(i, j) \in E$ by ant $t$,
is performed by
\begin{equation*}
Ph_{i, j}(k) = (1 - \delta) *  Ph_{i, j}(k) + \delta * Ph_{t}(k),
\end{equation*}
the level of the pheromone associated with color $k$ stored
on ant $t$ becomes the subject of the local update
\begin{equation*}
Ph_{t}(k) = Ph_{t}(k) - \delta * Ph_{t}(k),
\end{equation*}
where $\delta$ is a user defined evaporation parameter
such that $0 \le \delta \le 1$. It should be noted here that
the same notation has been employed to keep track of the
pheromone levels on both the edges, and the ants. However,
the levels of the pheromones stored on ants are tracked
with a single subscripted index as opposed to two for
the edges. 

The second type of the update to pheromone levels comes under
the title of the global update. The global pheromone update is also
known as off-line pheromone update. It is applied, at the end of each
iteration, only to the edges that are on the best path found so far.
The pheromone level for each color $k \in C$ on each such edge
is updated using the following formula
\begin{equation*}
Ph_{i, j}(k) = (1 - \delta) * Ph_{i, j}(k) + \delta * \frac{1}{cost(bestPath)}
\end{equation*}
where $cost(bestPath)$ denotes the total cost of traversing the
edges associated with the best, possibly non-simple, path found so far.

The pseudo-code of \emph{ACO-ACSP} heuristic algorithm, reflecting
the general anatomy of ant colony optimization as applied to
\emph{ACSP} is presented in Figure~\ref{fig:ant-colony-pseudo-code}.

\begin{figure}[h!]
\begin{algorithmic}[1]
\begin{small}
\Procedure{selectEdge(Ant ant)}{}
	\If{$\textit{ant.isDone}$ OR $\textit{ant.isDiscarded}$}
		\State $\textit{ant.pheromoneUpdate} \leftarrow \textit{false};$
	
	\Else
		\State $\textit{incidentEdges} \leftarrow \textit{findAvailableEdges}(\textit{ant});$
		\If{$\textit{sizeOf}(\textit{incidentEdges})=0$} 
			\State $\textit{ant.isDiscarded} \leftarrow \textit{true};$
		\Else
			\State $\textit{ant.pheromoneUpdate} \leftarrow \textit{true};$
			\State $\textit{sum} \leftarrow \textit{calcProbUsingPheromones}();$
			\If{$\textit{sum}=0$}
				\State \textit{calcProbUsingDistances}();
			\EndIf
			\State \textit{updateAntToPickEdge}();
		\EndIf
		
	\EndIf
\EndProcedure
\end{small}
\end{algorithmic}
\begin{algorithmic}[1]
\begin{small}
\Procedure {Aco-Acsp()}{}   
	\State \textit{bestAnt} $\gets \emptyset;$
	\For{$j \gets 0 \textbf{ to } \textit{iterationCount}$}
		\State \textit{colony} = \textit{createAnts}(\textit{colonySize});
		\State $\textit{antsDoneTour} \gets 0;$
		\While{$\textit{antsDoneTour} < \textit{colonySize}$}
			\ForAll{\textit{ant} \textbf{in} \textit{colony}}
				\State \textit{selectEdge}(\textit{ant});
				\State \textit{updateLastSelectedEdge}();
				\If {$\textit{ant.hasAllColors}$}
					\State $\textit{ant.isDone} \leftarrow \textit{true};$
					\State ++\textit{antsDoneTour};
				\EndIf
			\EndFor
		\EndWhile
		\State $\textit{tempAnt} \leftarrow \textit{findBestAnt}();$
		\State \textit{updatePheromoneOnBestPath}(\textit{tempAnt.path});
		\If{$\textit{tempAnt.cost} \le \textit{bestAnt.cost}$}
			\State $\textit{bestAnt} \leftarrow \textit{tempAnt};$
		\EndIf
	\EndFor
	\State \textbf{return} \textit{bestAnt.cost};
\EndProcedure
\end{small}
\end{algorithmic}
\caption{\emph{ACO-ACSP}: Heuristic based on ant colony optimization.}
\label{fig:ant-colony-pseudo-code}
\end{figure}

\subsection{Genetic Algorithm}\label{sec:genetic}	

The Genetic Algorithm (GA) \cite{Holland} has five main steps: initialization,
fitness, selection, crossover, and mutation. In this section, we develop
\emph{GA-ACSP} which is another heuristic solution to \emph{ACSP} based
on GA. The algorithm is presented in Figure \ref{fig:genetic}.

During the initialization step,
a pool of chromosomes, called population, is generated. We encode the
chromosomes in such a way that each chromosome is represented by an
ordered list of vertices, corresponding to a solution to a given \emph{ACSP}
instance. As the path is not necessarily a simple path, each vertex may appear
multiple times on a chromosome. Initially, the population is filled with a certain
number of randomly created, possibly non-simple, paths, each of which visits
all the distinct colors at least once.

\begin{figure}[h!]
\begin{algorithmic}[1]
\begin{small}
\Procedure {Ga-Acsp}{}
	\State $\textit{population} \gets \emptyset;$
	\For {$i \gets 0 \textbf{ to } \textit{populationSize}$}
		\State $\textit{chrm} = \textit{createRandomChromosome}();$
		\State $\textit{ensureConnectivity}(\textit{chrm});$\label{GA:ensure1}
		\State $\textit{population.add}(\textit{chrm});$
	\EndFor
	\For {$i \gets 0 \textbf{ to } \textit{iterationCount}$}
		\State $\textit{candidates} \gets \textit{rouletteWheelSelection}(\textit{population});$
		\State $\textit{children} \gets \textit{crossOver}(\textit{candidates});$
		\ForAll {\textit{child} \textbf{in} \textit{children}}
			\State $r \gets \textit{random}(0,1);$
			\If{$r < \textit{mutationProbability}$}
				\State $\textit{mutate}(\textit{child});$
			\EndIf
			\State $\textit{completeMissingColors}(\textit{child});$
			\State $\textit{ensureConnectivity}(\textit{child});$\label{GA:ensure2}
			\State $\textit{population.add}(\textit{child});$
		\EndFor
		\State $\textit{w2c} \gets \textit{find2chromosomesWithLowestFitness}();$
		\State $\textit{population.remove}(\textit{w2c});$
	\EndFor
	\State \textbf{return} \textit{costOfBestChromosome};
\EndProcedure
\end{small}
\end{algorithmic}
\caption{\emph{GA-ACSP}: Heuristic based on genetic algorithm.}
\label{fig:genetic}
\end{figure}

In the next step, the fitness values are calculated for each chromosome
in the population. The fitness value of a chromosome, in \emph{GA-ACSP},
is simply taken as the total distance traveled down the corresponding possibly
non-simple path.

In the selection step, two candidate chromosomes are selected from the
population for crossover. The selection of the candidates are performed
using the roulette wheel selection algorithm \cite{Goldberg}, in which
the chromosomes with higher fitness values have higher chances to be selected. 

In performing a crossover, two random positions $p_1$, and $p_2$ with a common
vertex are initially figured out in the first, and the second candidate chromosomes
respectively. The portions beyond $p_1$, and $p_2$ are then swapped between the
candidates to produce two new children. In case the candidates do not have a common
vertex, two more candidates for crossover are selected until a vertex common
to both can be found. In the end, two new, possibly non simple, paths are generated.
It should be noted, however, that these new paths are not guaranteed to visit all the
colors. Therefore, as soon as the crossover, and the mutation steps are over, we examine
the paths to find a list of the missing colors on each, and then, modify the paths accordingly
so that when the process is over, each of the two new chromosomes represents also a
non-simple path that visits each color at least once. In modifying the paths, we follow
a greedy policy, and append to the tail of the chromosome, at each iteration, the shortest
path from the tail to the closest vertex with a color not visited yet. This helps
keeping the path lengths as short as possible.

The last step in \emph{GA-ACSP} is the mutation. 
It is carried out by simply replacing two random vertices in the chromosome,
and rearranging the path to ensure that it remains connected. In order to
connect non-neighbor vertices in the chromosome, the shortest path between
those vertices is inserted into the chromosome.

After the crossover, and the mutation steps, the child chromosomes
need to be verified to correctly represent a possibly non-simple path.
After any possible problems regarding missing colors, and disconnectivity
are dealt with, as described above, there is still some more work to do.
First, the redundant segment at the tail part of the chromosome should be
cropped if all the colors have already been seen before the beginning
of that segment. Lastly, the property that no edge can be visited
more than once in any direction in an optimal solution can be violated
as a result of the crossover, and the mutation operations.
In such a case, the property should be restored, as highlighted
in the proof to Proposition~\ref{proposition:maxOneVisit}.

Once the two newly formed child chromosomes are added to the population,
the two chromosomes that have the worst fitness values are removed from
the population.

\section{Experiments}\label{sec:experiments}

In this section, we present the results of our experiments for the proposed heuristic
algorithms. We refer to
\emph{$\text{LP}_{\text{x}}\text{ACSP}$},
\emph{$\text{LP}_{\text{f}}\text{ACSP}$}, and
\emph{$\text{LP}_{\text{f/x}}\text{ACSP}$}
under the heading of LP relaxation based algorithms whereas
\emph{SA-ACSP}, \emph{ACO-ACSP}, and \emph{GA-ACSP}
are treated under the category of metaheuristic algorithms.
We implemented the metaheuristic algorithms in C++, and used
CPLEX for ILP, and LP relaxation based heuristics.
All tests are performed on computers that have AMD
Phenom(tm) II X4 810 2.67 GHz CPU, and 2 Gb 400 MHz DDR2 RAM running
on the 32-bit operating system Ubuntu 10.04. We conducted the experiments on
randomly generated graphs with varying number of nodes, and colors
as listed in Table \ref{table:nonlin}, with an average node degree of $6$,
uniform color distribution, and an average edge weight of $10$. The simulations are
conducted $10$ times for each graph type, and only the average, and the minimum
cost values are reported.

\begin{table}[h!]
\centering
\begin{tabular}{c c c}
\hline\hline
Graph Name & Number of Nodes & Number of Colors \\ [0.5ex]
\hline
n50-c10 & 50 & 10 \\
n50-c20 & 50 & 20 \\
n50-c25 & 50 & 25 \\
n100-c25 & 100 & 25 \\
n100-c40 & 100 & 40 \\
n100-c50 & 100 & 50 \\
n200-c50 & 200 & 50 \\
n200-c75 & 200 & 75 \\[1ex]

\hline
\end{tabular}
\caption{The number of nodes, and colors for the randomly generated graph types
used in the experiments.}
\label{table:nonlin}
\end{table}

In the sections to follow, the experimental results are reported, first, separately
for each metaheuristic algorithm. In each of these sections, we conduct various
experiments for parameter tuning, namely, to find the optimal values of individual
parameters specific to a metaheuristic. Finally, in Section \ref{sec:comparison},
an overall comparison is presented to assess the relative performance of all the
heuristics proposed, using the fine-tuned parameters.

\subsection{\emph{SA-ACSP}: Parameter Tuning for SA}
In SA, the two parameters essential to performance are the best cooling
rate with respect to time and cost, and the best temperature to be used
with the best cooling rate. All the other parameters for \emph{SA-ACSP}
are kept unchanged during these tests.

The cooling rate is used to determine the amount of decrease in the temperature
value at each iteration. 
We tested \emph{SA-ACSP} with various cooling rate values as
presented in Figures \ref{Coolingrate-Cost-Min}, \ref{Coolingrate-Cost-Avg}, and
\ref{Coolingrate-Time-Avg}. The temperature value for these tests is set at
$100$. Looking at Figures \ref{Coolingrate-Cost-Min}, and
\ref{Coolingrate-Cost-Avg}, we can observe that the cost slightly decreases with
the increasing values of cooling rate parameter. In Figure
\ref{Coolingrate-Time-Avg}, however, we also observe that the time increases dramatically
for the cooling rate values larger than $0.999$. Therefore, based on the results of
these experiments, we selected the best cooling rate parameter to be $0.999$.

\begin{figure} [h!]
\centering
\includegraphics[width=\columnwidth] {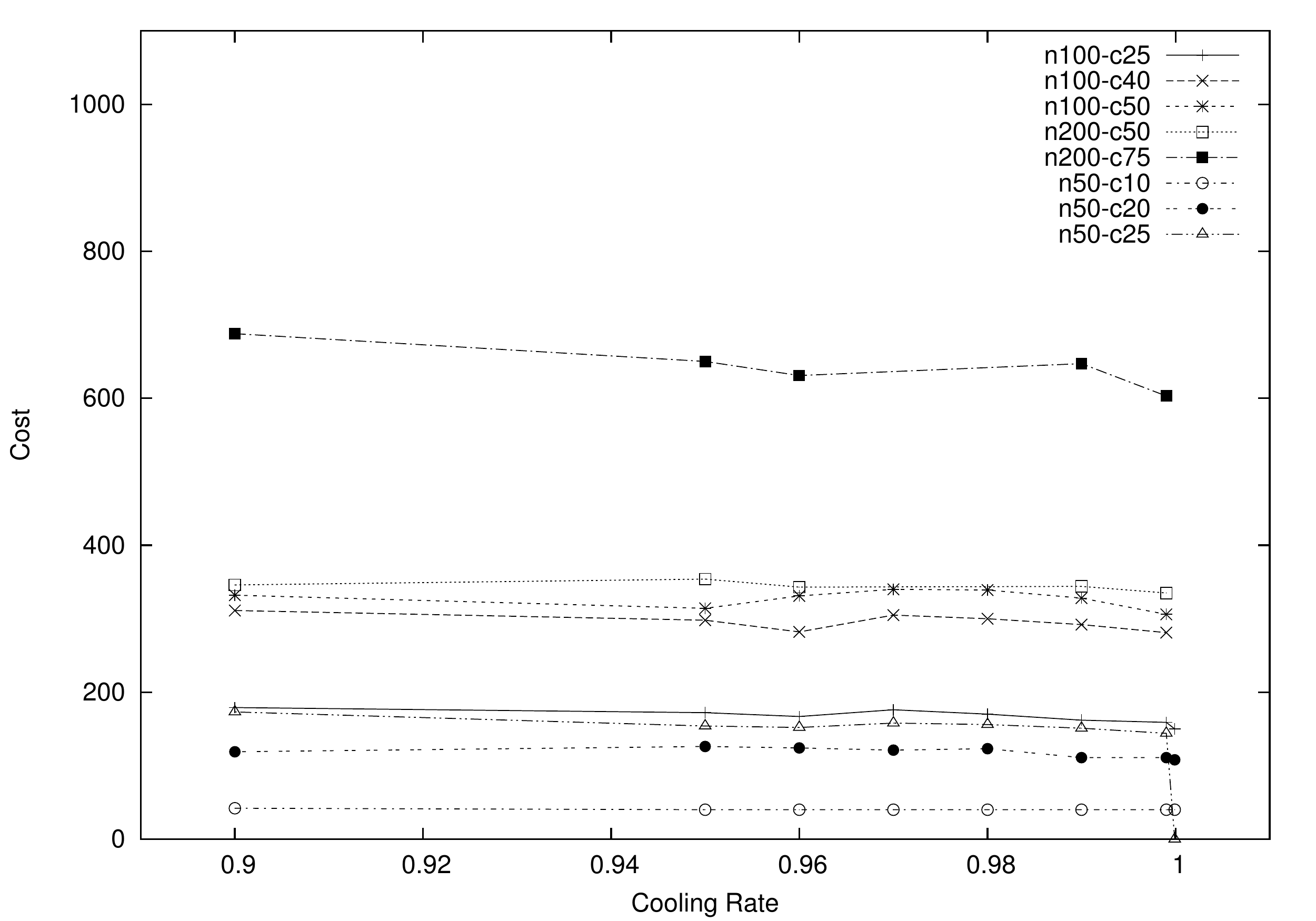}
\caption{Minimum cost for various cooling rate values in \emph{SA-ACSP}.}
\vspace{-2ex}
\label{Coolingrate-Cost-Min}
\end{figure}

\begin{figure} [h!]
\centering
\includegraphics[width=\columnwidth] {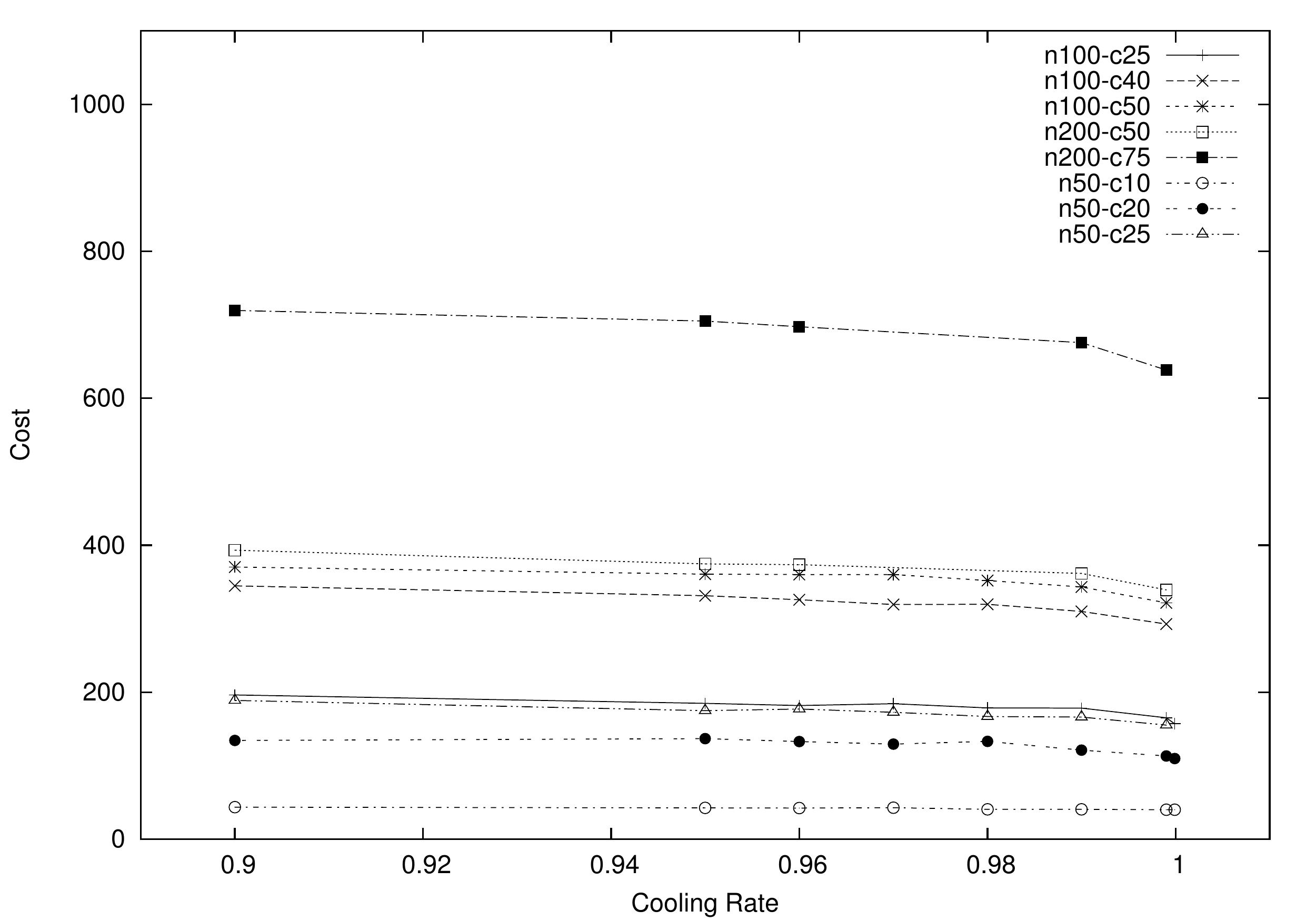}
\caption{Average cost for various cooling rate values in \emph{SA-ACSP}.}
\label{Coolingrate-Cost-Avg}
\end{figure}

\begin{figure} [h!]
\centering
\includegraphics[width=\columnwidth] {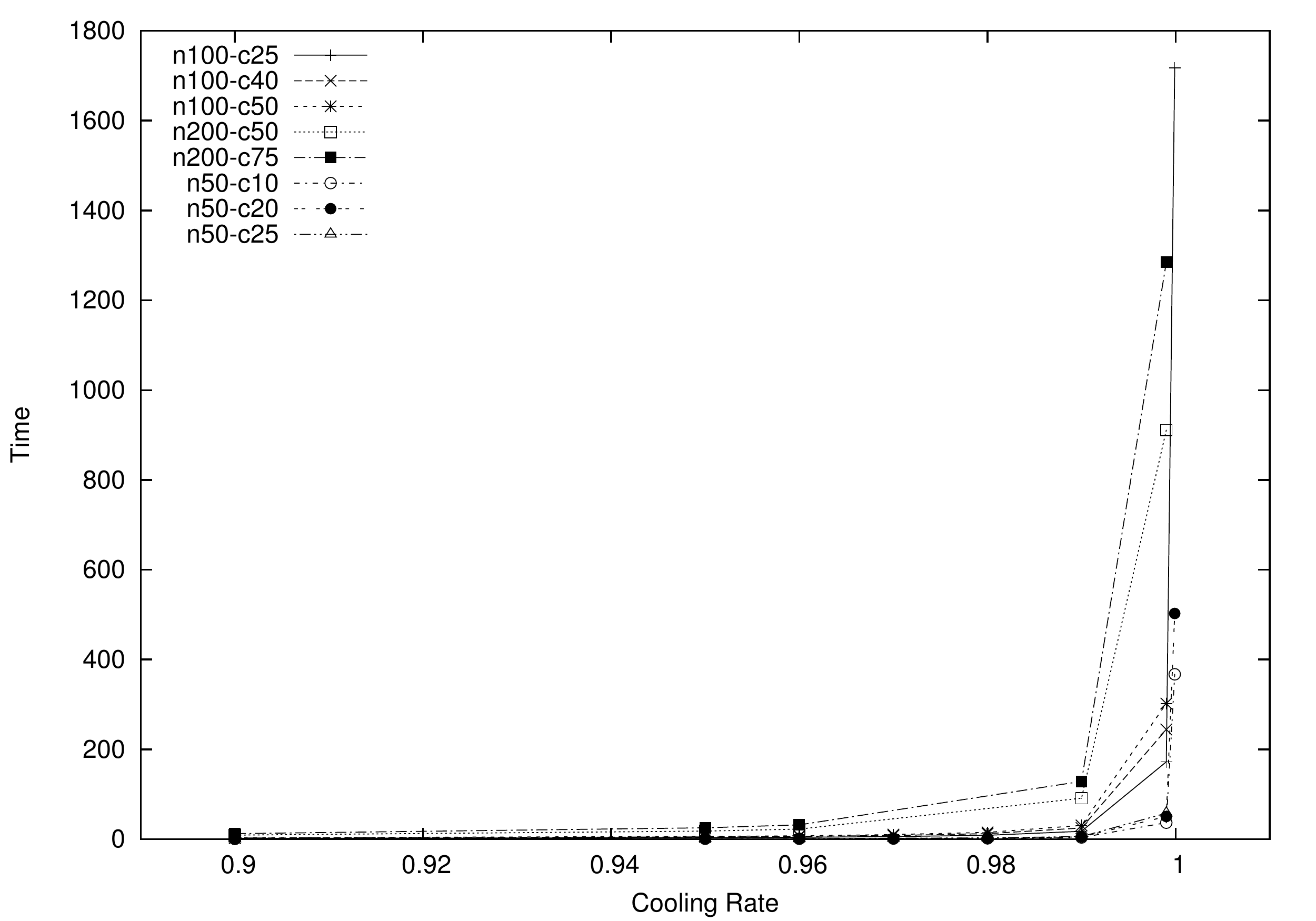}
\caption{CPU time for calculating the average cost values for various cooling rate
values in \emph{SA-ACSP}.}
\label{Coolingrate-Time-Avg}
\end{figure}

The temperature value in \emph{SA-ACSP} controls the probability of
choosing worse paths in order to not get stuck at a local minimum.
We conducted simulations with various temperature values
on all graph types, and present the cost, and the CPU times in Figures
\ref{Temperature-Cost-Min}, \ref{Temperature-Cost-Avg}, and
\ref{Temperature-Time-Avg}. Based on the results in these figures, we selected
$1000$ as the best temperature value.

\begin{figure} [h]
\centering
\includegraphics[width=\columnwidth] {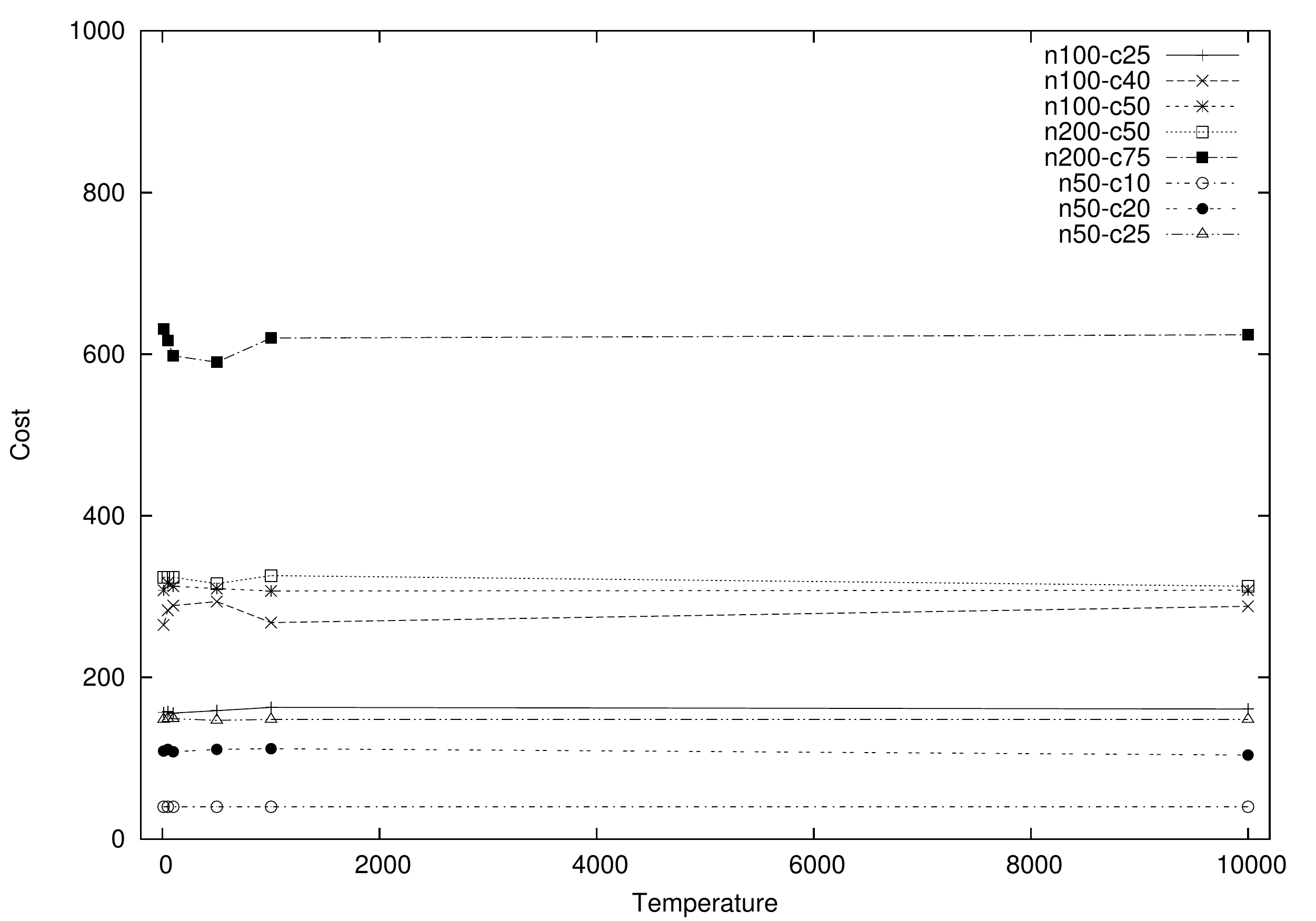}
\caption{Minimum cost for various temperature values in \emph{SA-ACSP}.}
\label{Temperature-Cost-Min}
\end{figure}

\begin{figure} [h]
\centering
\includegraphics[width=\columnwidth] {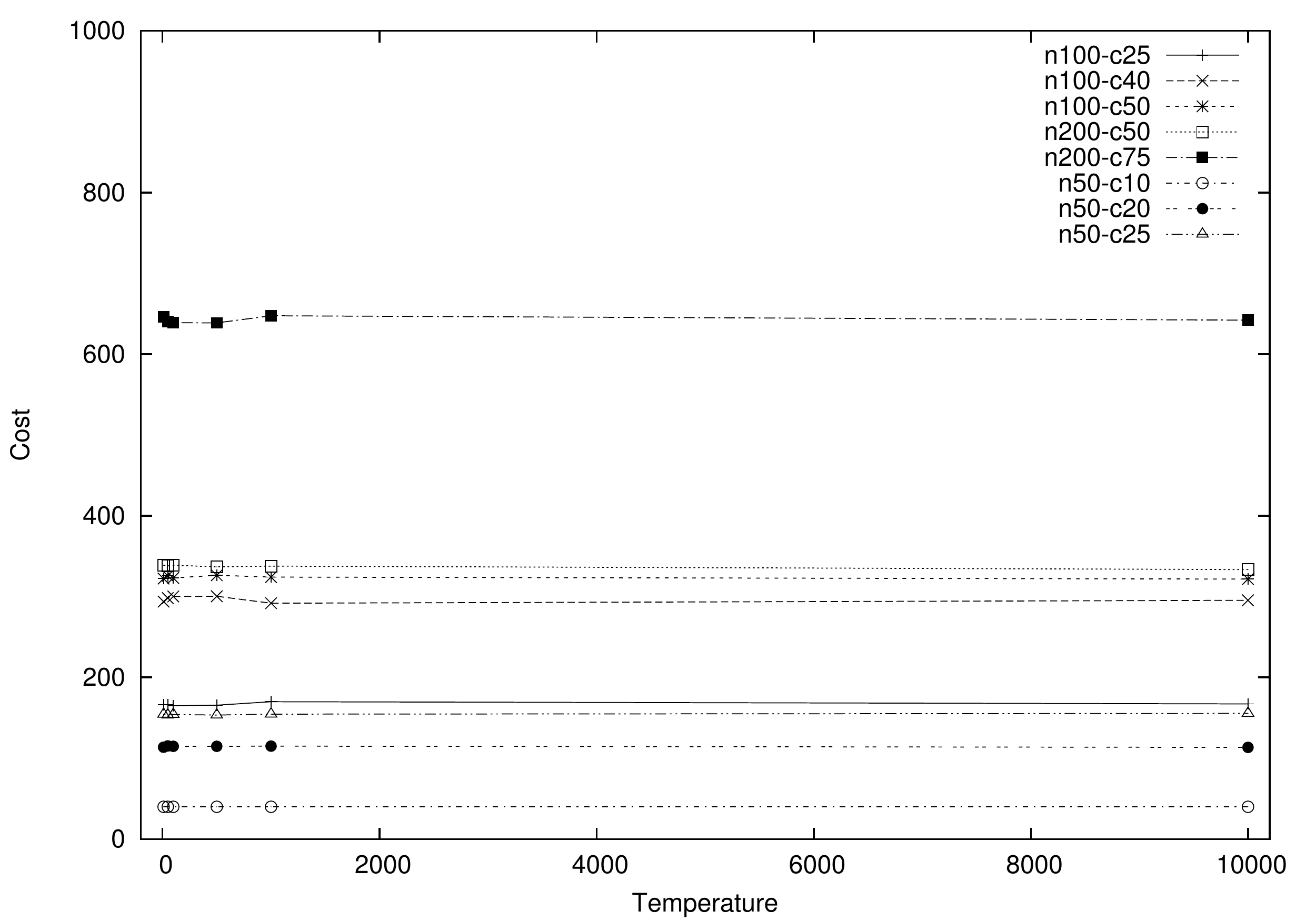}
\caption{Average cost for various temperature values in \emph{SA-ACSP}.}
\label{Temperature-Cost-Avg}
\end{figure}

\begin{figure} [h]
\centering
\includegraphics[width=\columnwidth] {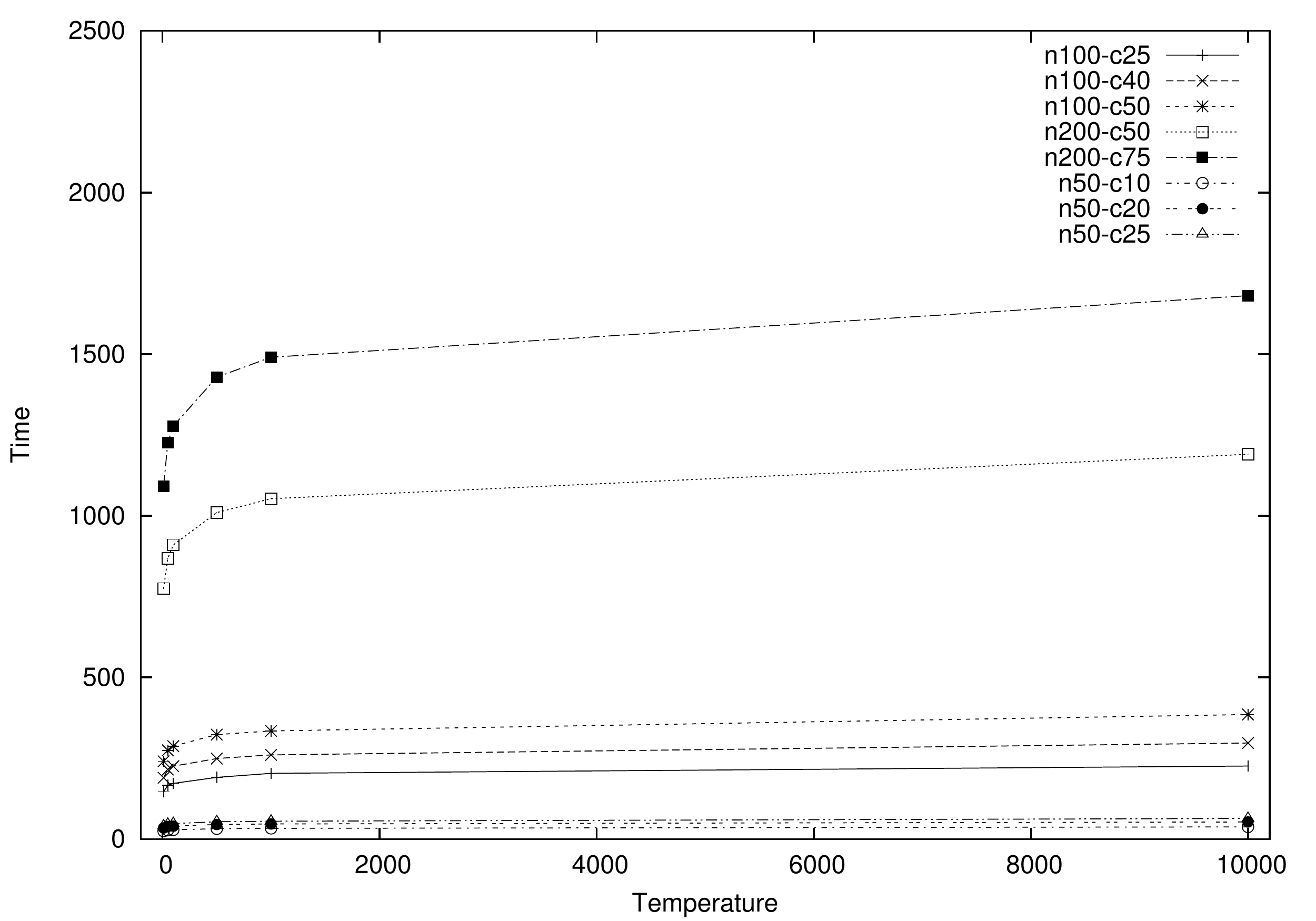}
\caption{CPU time for calculating average cost for various temperature values in \emph{SA-ACSP}.}
\label{Temperature-Time-Avg}
\end{figure}

\subsection{\emph{ACO-ACSP}: Parameter Tuning for ACO}

The behavior of \emph{ACO-ACSP} depends on four separate
parameters. In order to find the optimal value for each parameter, we conducted
a series of experiments for each individual parameter. In each experiment, all
the other parameters are kept constant, and only the specified parameters are
tested for various values, and the cost, and the runtime values observed
are recorded. The experiments conducted can be categorized into alpha-beta
tests, colony size tests, and probability tests.

Alpha, and beta values are the two parameters used for edge selection in
\emph{ACO-ACSP}. The parameter alpha denotes the importance of the
pheromone levels on the edges while calculating probabilities. Beta,
on the other hand, represents the importance of the edge weights.
The experiments are designed to decide on the combination of alpha, and beta
values that gives us the best result in terms of the cost, and the CPU time.
The results of the experiments are presented in Figures~\ref{Alphabeta-Cost-Min},
\ref{Alphabeta-Cost-Avg}, and \ref{Alphabeta-time-Avg}. These figures report the
results with respect to the CPU time, the average cost, and the minimum cost, and
based on the results, we selected the alpha value as $0.4$, and beta value as $0.5$.

\begin{figure} [h!]
\centering
\includegraphics[width=\columnwidth] {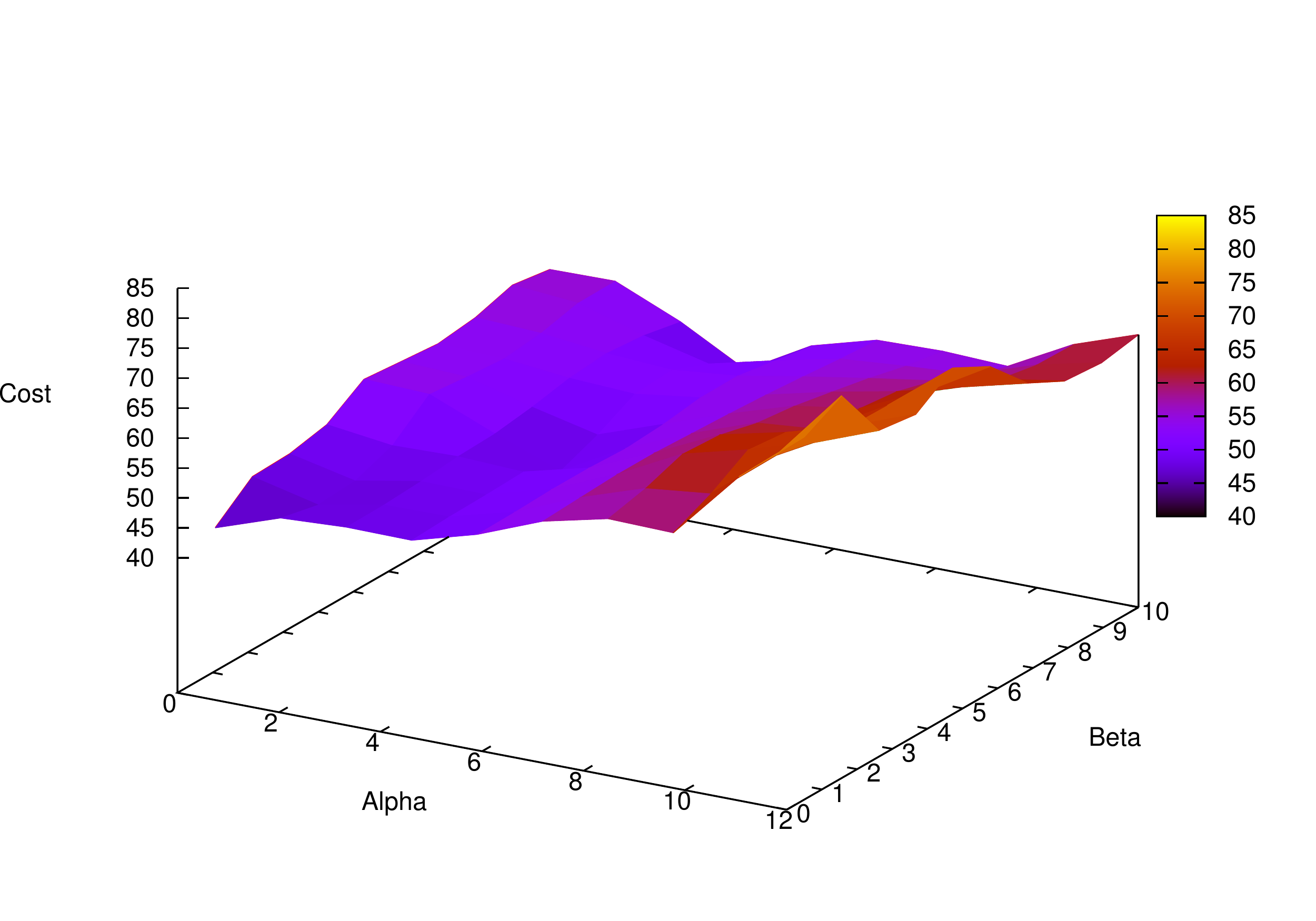}
\caption{Minimum cost for various combinations of alpha, and beta values in \emph{ACO-ACSP}.}
\label{Alphabeta-Cost-Min}
\end{figure}

\begin{figure} [h!]
\centering
\includegraphics[width=\columnwidth] {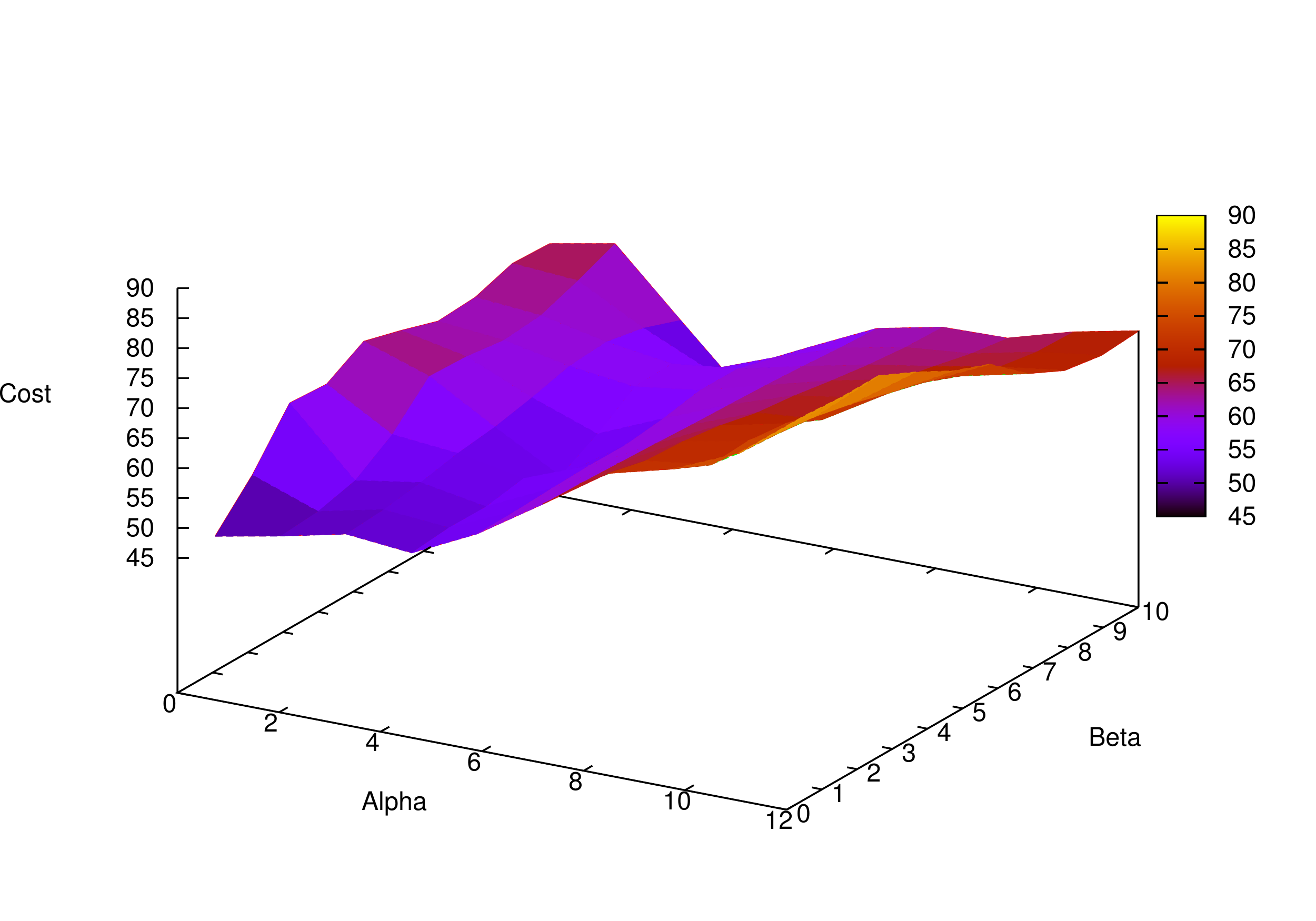}
\caption{Average cost for various combinations of alpha, and beta values in \emph{ACO-ACSP}.}
\label{Alphabeta-Cost-Avg}
\end{figure}

\begin{figure} [h!]
\centering
\includegraphics[width=\columnwidth] {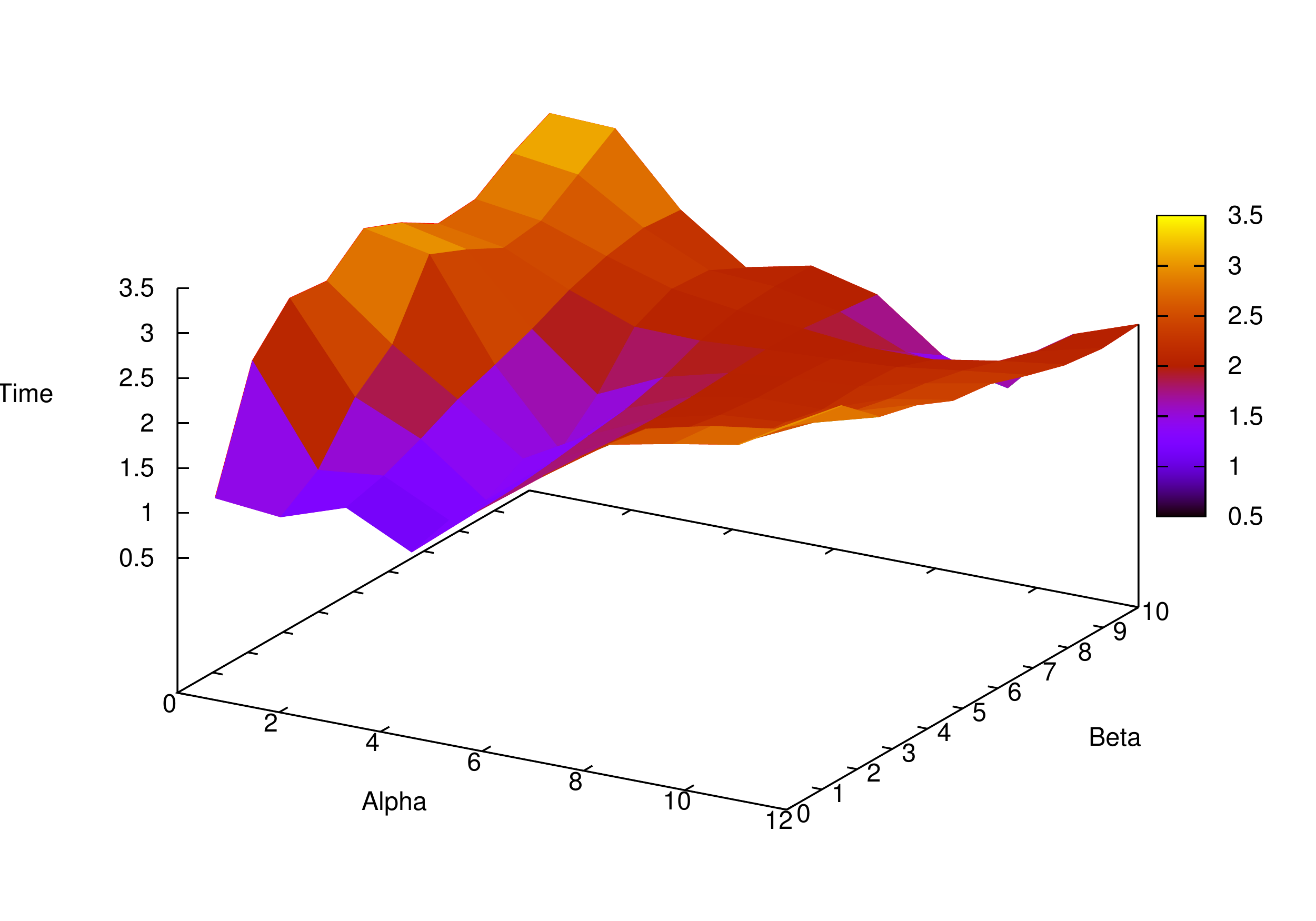}
\caption{CPU time for calculating the average cost for various combinations of alpha, and beta
values in \emph{ACO-ACSP}.}
\label{Alphabeta-time-Avg}
\end{figure}

In \emph{ACO-ACSP}, the colony size represents the number of active ants
deployed at each iteration of the algorithm. Having a larger colony increases the
chances for finding solutions closer to the optimal, however, at the cost of increasing
the overall runtime. Therefore, we test \emph{ACO-ACSP} for various colony sizes to
decide on the optimal colony size that can achieve a minimal cost solution in an acceptable
time period. The results of the experiments are presented in Figures~\ref{Colonysize-Cost-Min},
\ref{Colonysize-Cost-Avg}, and \ref{Colonysize-Time-Avg}.
As we can clearly see from the figures, the runtime increases linearly in the size of the colony,
and the cost decreases only slightly for colony sizes larger than $200$. Based on these results,
we selected the colony size as $200$ in the rest of the experiments.

\begin{figure} [h!]
\centering
\includegraphics[width=\columnwidth] {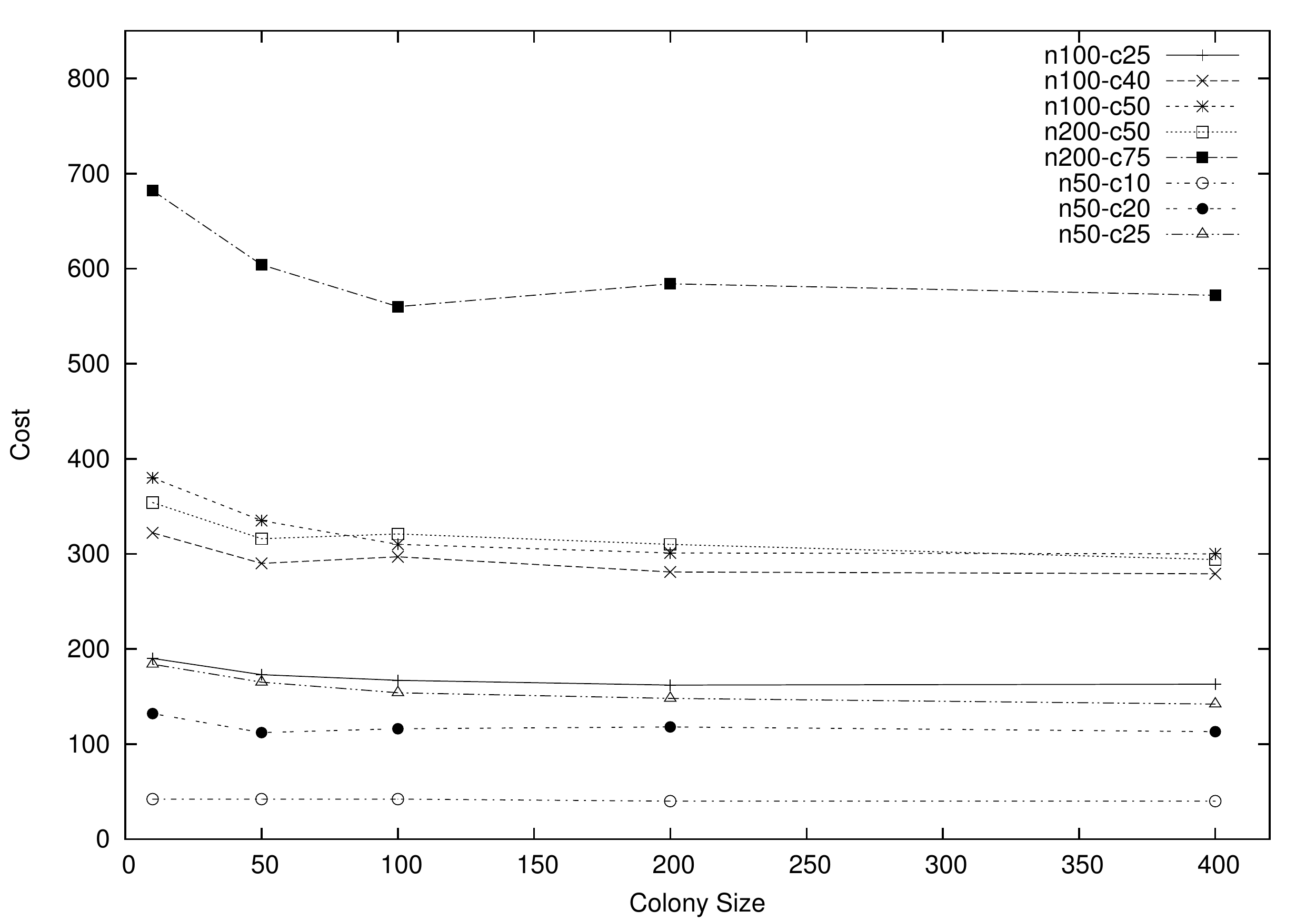}
\caption{Minimum cost for various colony size values in \emph{ACO-ACSP}.}
\label{Colonysize-Cost-Min}
\end{figure}

\begin{figure} [h!]
\centering
\includegraphics[width=\columnwidth] {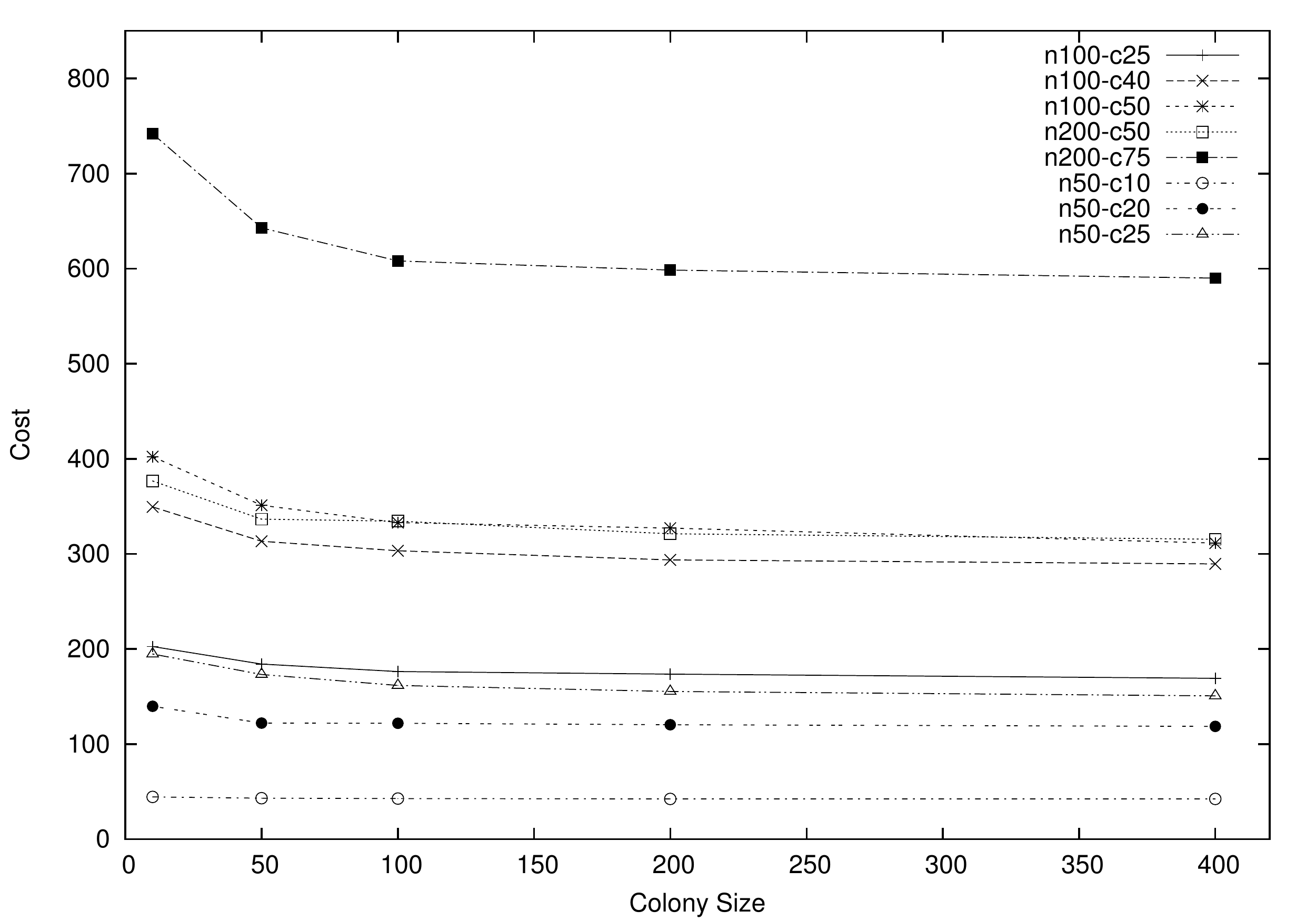}
\caption{Average cost for various colony size values in \emph{ACO-ACSP}.}
\label{Colonysize-Cost-Avg}
\end{figure}

\begin{figure} [h!]
\centering
\includegraphics[width=\columnwidth] {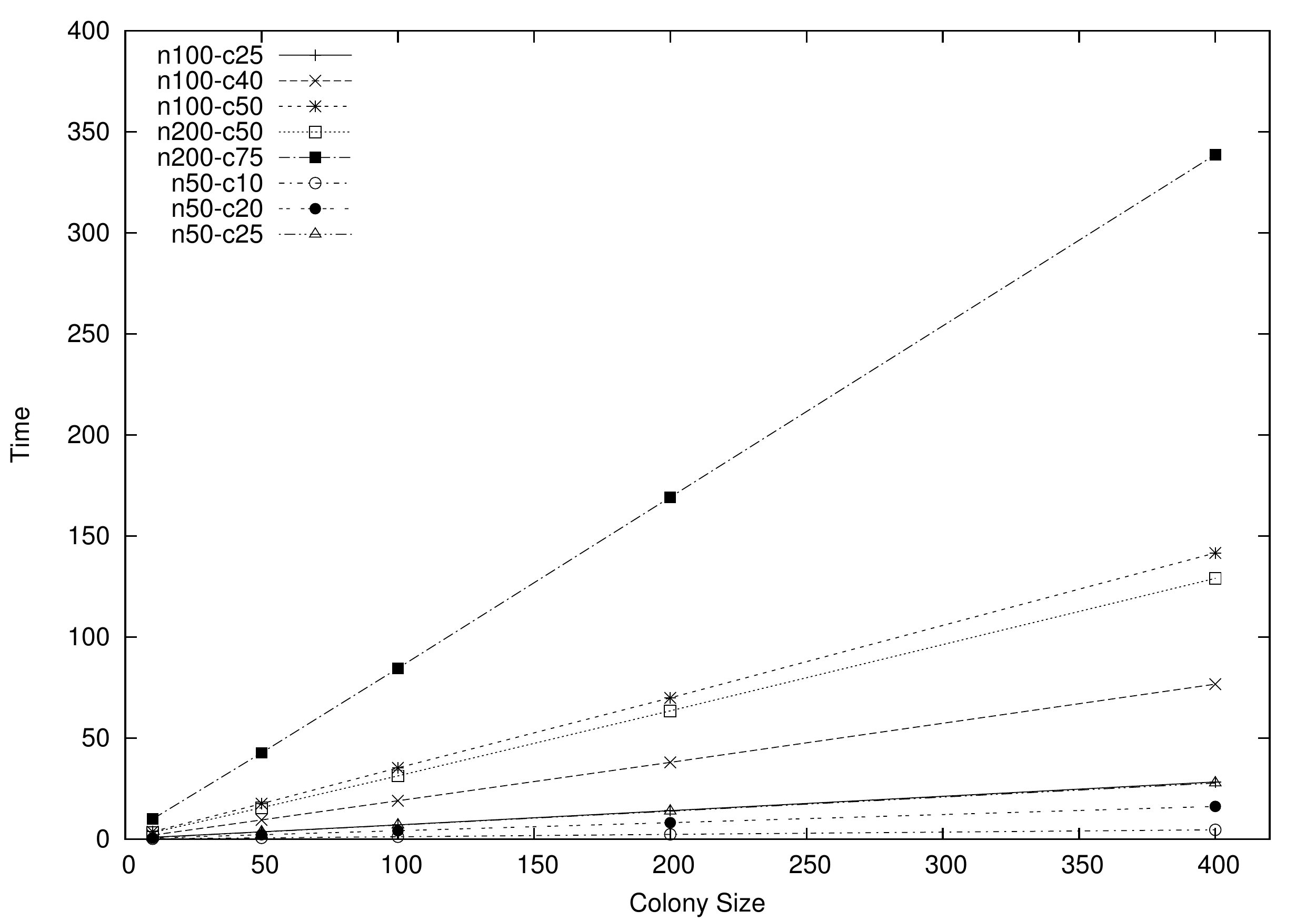}
\caption{CPU time for calculating the average cost for various colony sizes in \emph{ACO-ACSP}.}
\label{Colonysize-Time-Avg}
\end{figure}

The edge selection probabilities in \emph{ACO-ACSP} are calculated based on
either the level of pheromones on the edge or the edge weight itself.
The results of the experiments conducted to find the optimal probability value
is presented in Figures~\ref{Probability-Cost-Min}, \ref{Probability-Cost-Avg},
and \ref{Probability-Time-Avg}. A close inspection of these figures reveal that
both the cost, and the runtime increase for probability values larger than $0.95$.
Based on this observation, therefore, we selected the probability value as $0.9$
to be used throughout the rest of the experiments.

\begin{figure} [h!]
\centering
\includegraphics[width=\columnwidth] {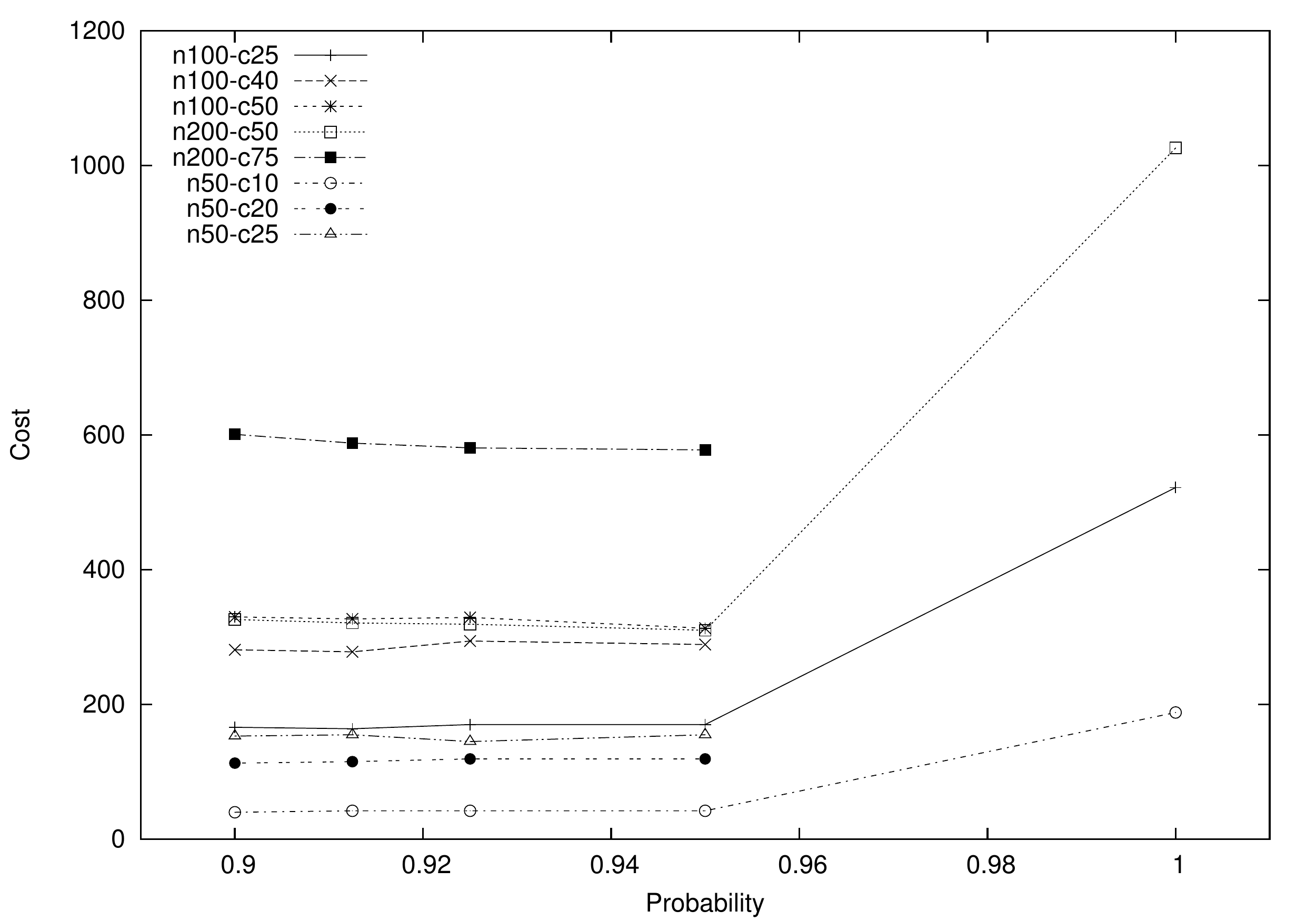}
\caption{Minimum cost for various edge selection probability values in \emph{ACO-ACSP}.}
\label{Probability-Cost-Min}
\end{figure}

\begin{figure} [h!]
\centering
\includegraphics[width=\columnwidth] {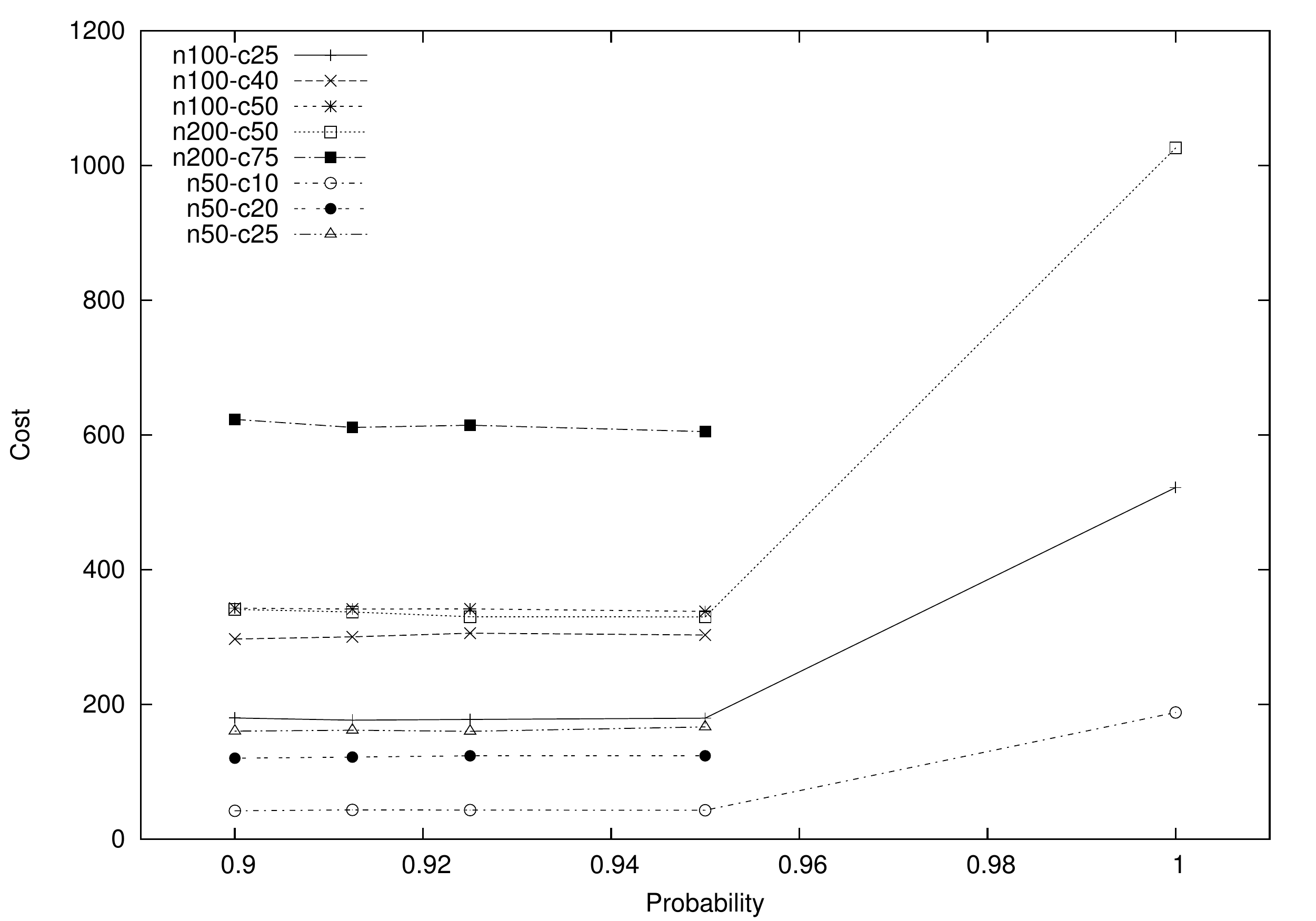}
\caption{Average cost for various edge selection probability values in \emph{ACO-ACSP}.}
\label{Probability-Cost-Avg}
\end{figure}

\begin{figure} [h!]
\centering
\includegraphics[width=\columnwidth] {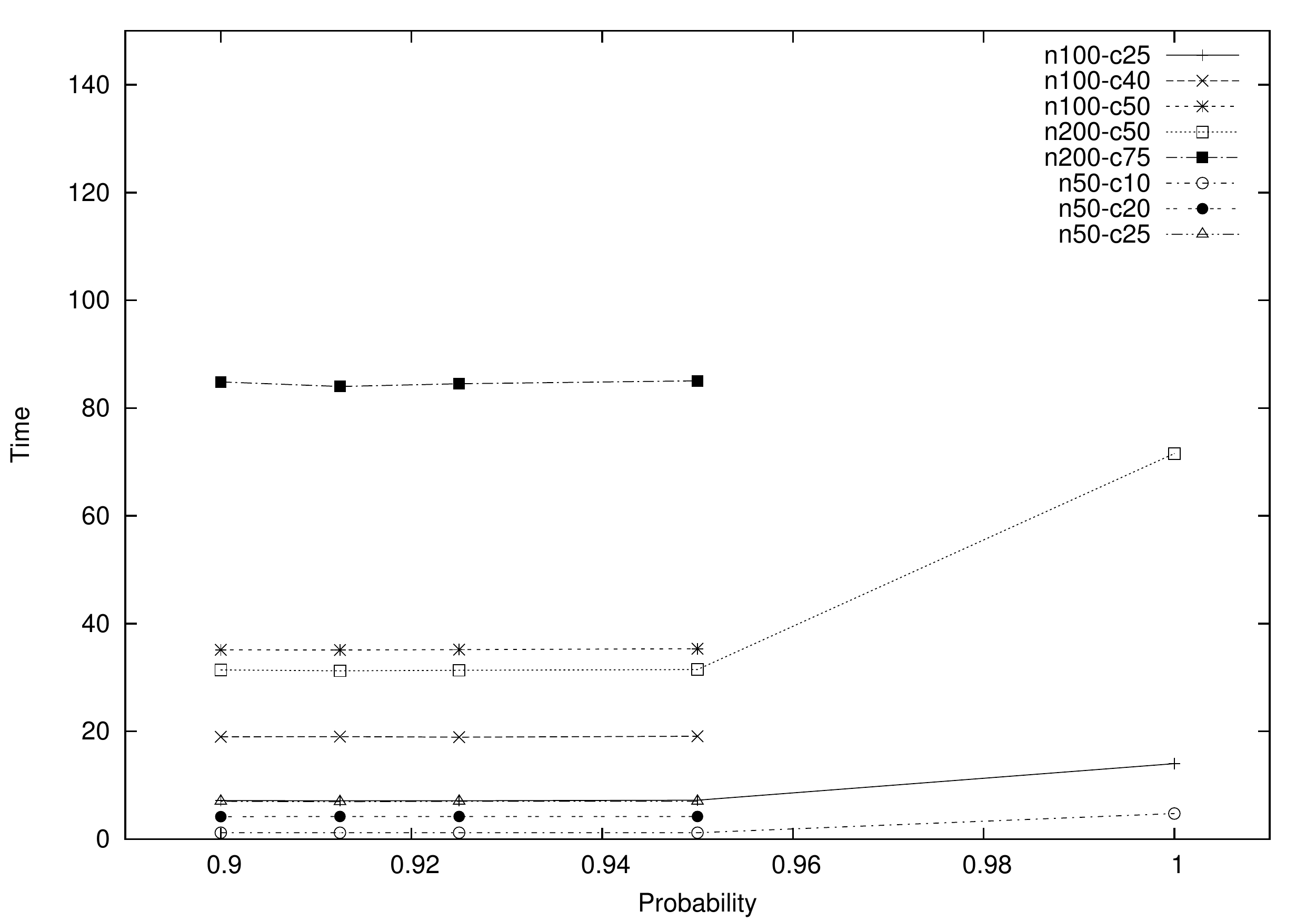}
\caption{CPU time for calculating the average cost for various edge selection
probability values in \emph{ACO-ACSP}.}
\label{Probability-Time-Avg}
\end{figure}

\subsection{\emph{GA-ACSP}: Parameter Tuning for GA}

In \emph{GA-ACSP}, the parameters investigated are the mutation
probability, the population size, and the iteration size. The results of the
tests on the mutation probability are presented in Figures~\ref{mutation-cost-min},
\ref{mutation-cost-avg}, and \ref{mutation-time-avg}. It can be seen from the
figures that the changes on the mutation probability do not affect the cost
dramatically. Therefore, to prevent over-randomization of chromosomes,
we selected for the mutation probability a value as low as $0.1$. 

\begin{figure} [h!]
\centering
\includegraphics[width=\columnwidth] {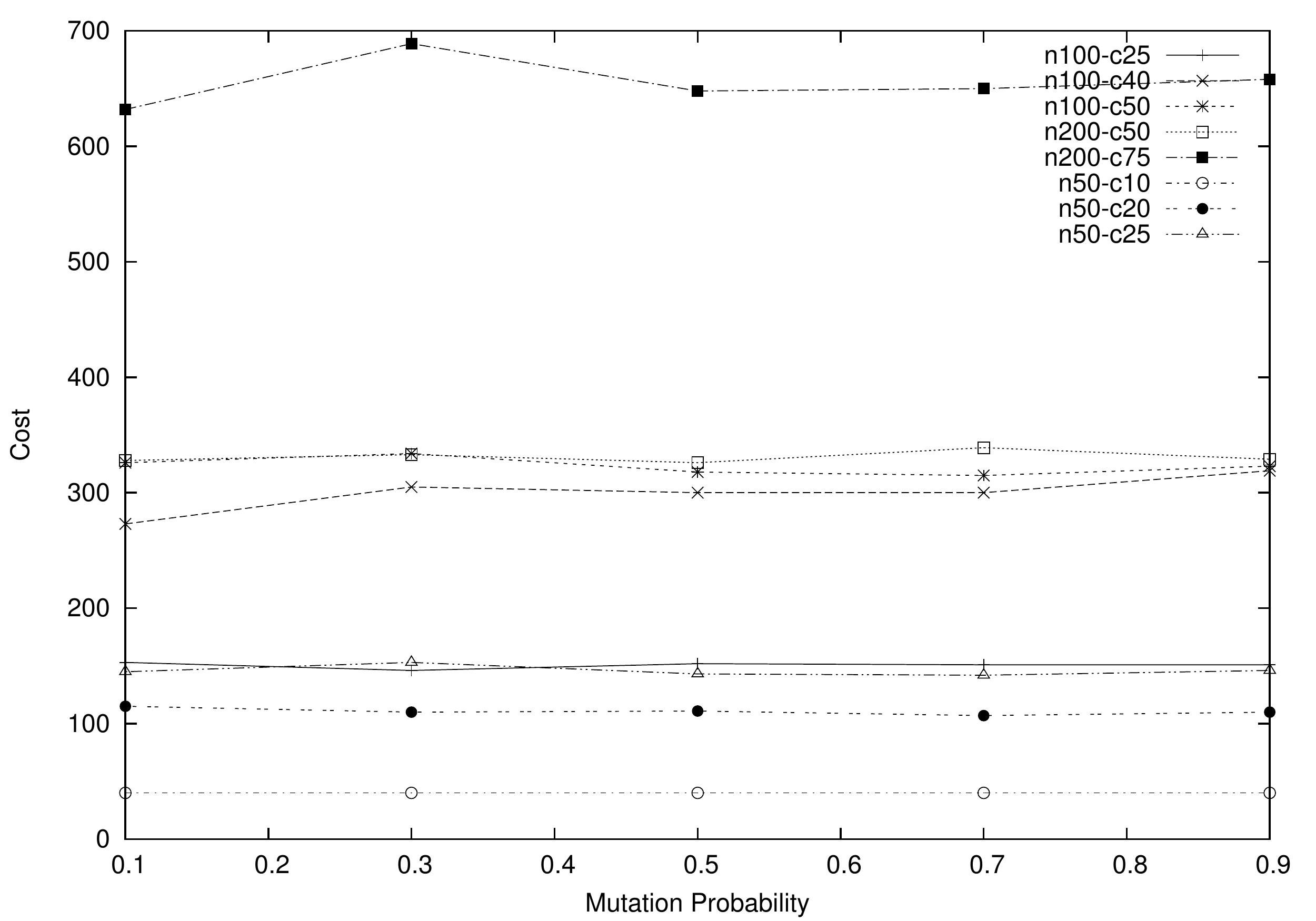}
\caption{Minimum cost for various mutation probability values in \emph{GA-ACSP}.}
\label{mutation-cost-min}
\end{figure}

\begin{figure} [h!]
\centering
\includegraphics[width=\columnwidth] {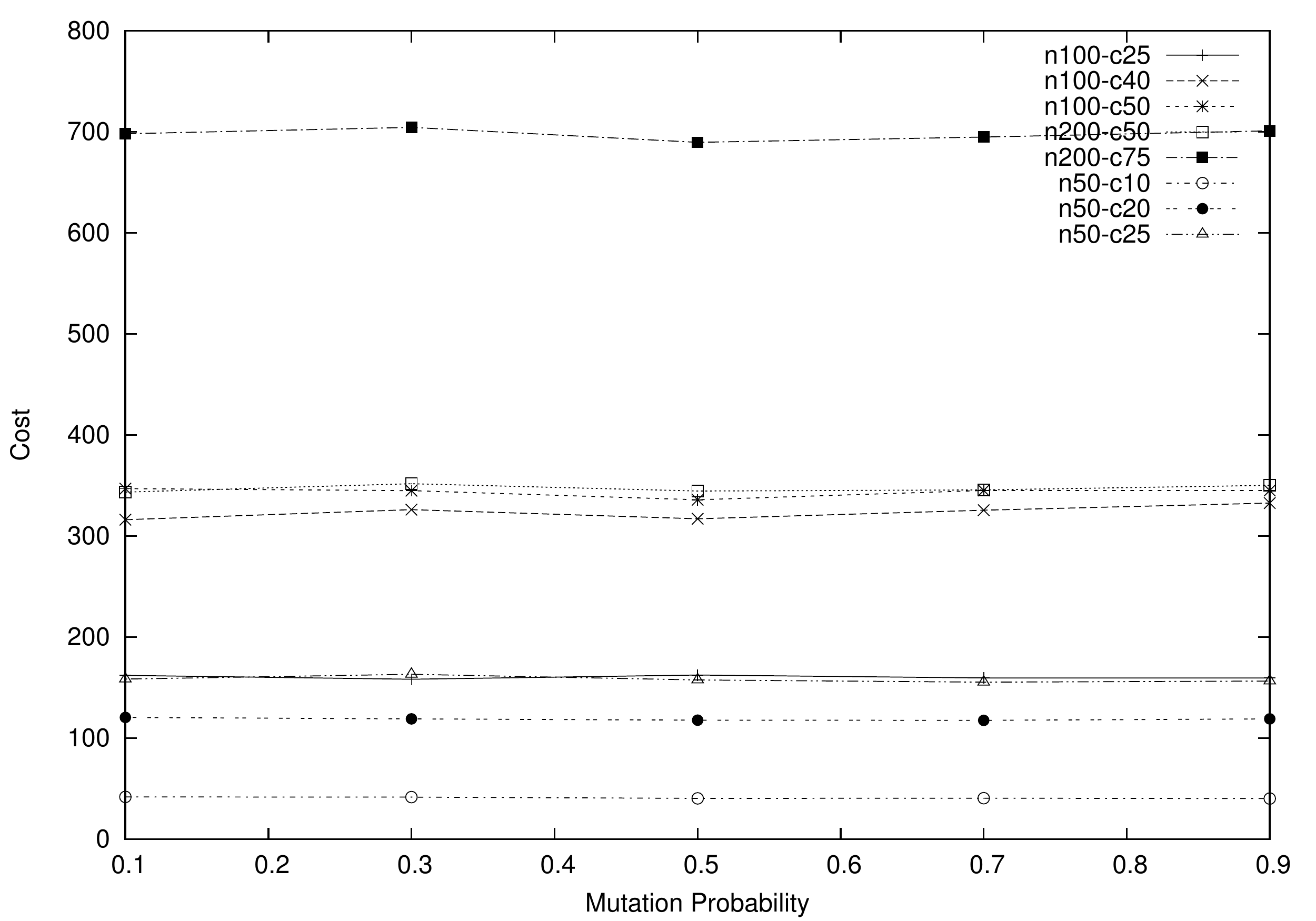}
\caption{Average cost for various mutation probability values in \emph{GA-ACSP}.}
\label{mutation-cost-avg}
\end{figure}

\begin{figure} [h!]
\centering
\includegraphics[width=\columnwidth] {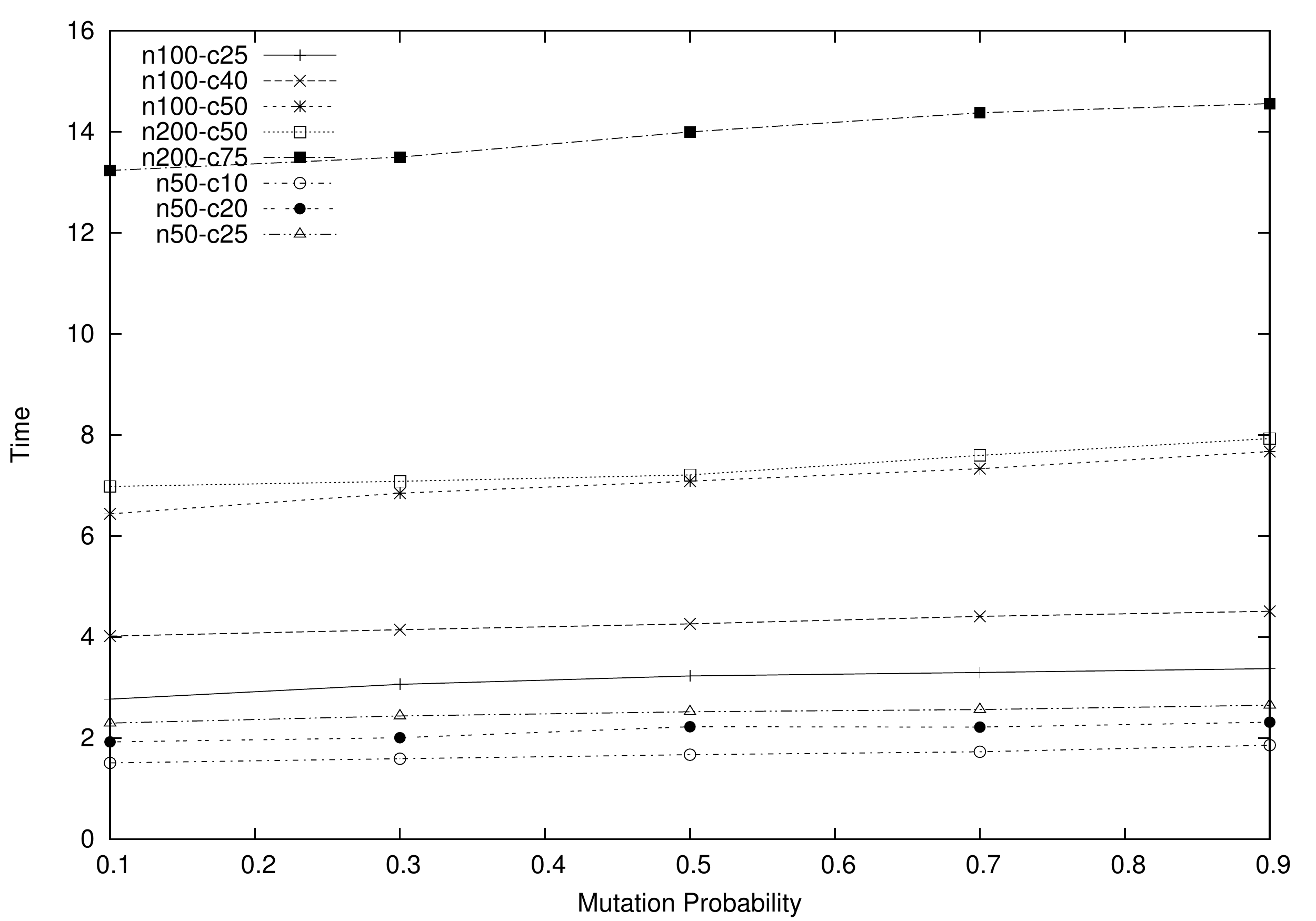}
\caption{CPU time for calculating the average cost for various mutation probability
values in \emph{GA-ACSP}.}
\label{mutation-time-avg}
\end{figure}

We also conducted experiments to find the optimal iteration size. The results of the
experiments are presented in Figures \ref{itersize-cost-min},
\ref{itersize-cost-avg}, and \ref{itersize-time-avg}. Based on the results,
we decided to select the iteration size as $6000$, as it provides the best trade-off
between the cost, and the runtime. 

\begin{figure} [h!]
\centering
\includegraphics[width=\columnwidth] {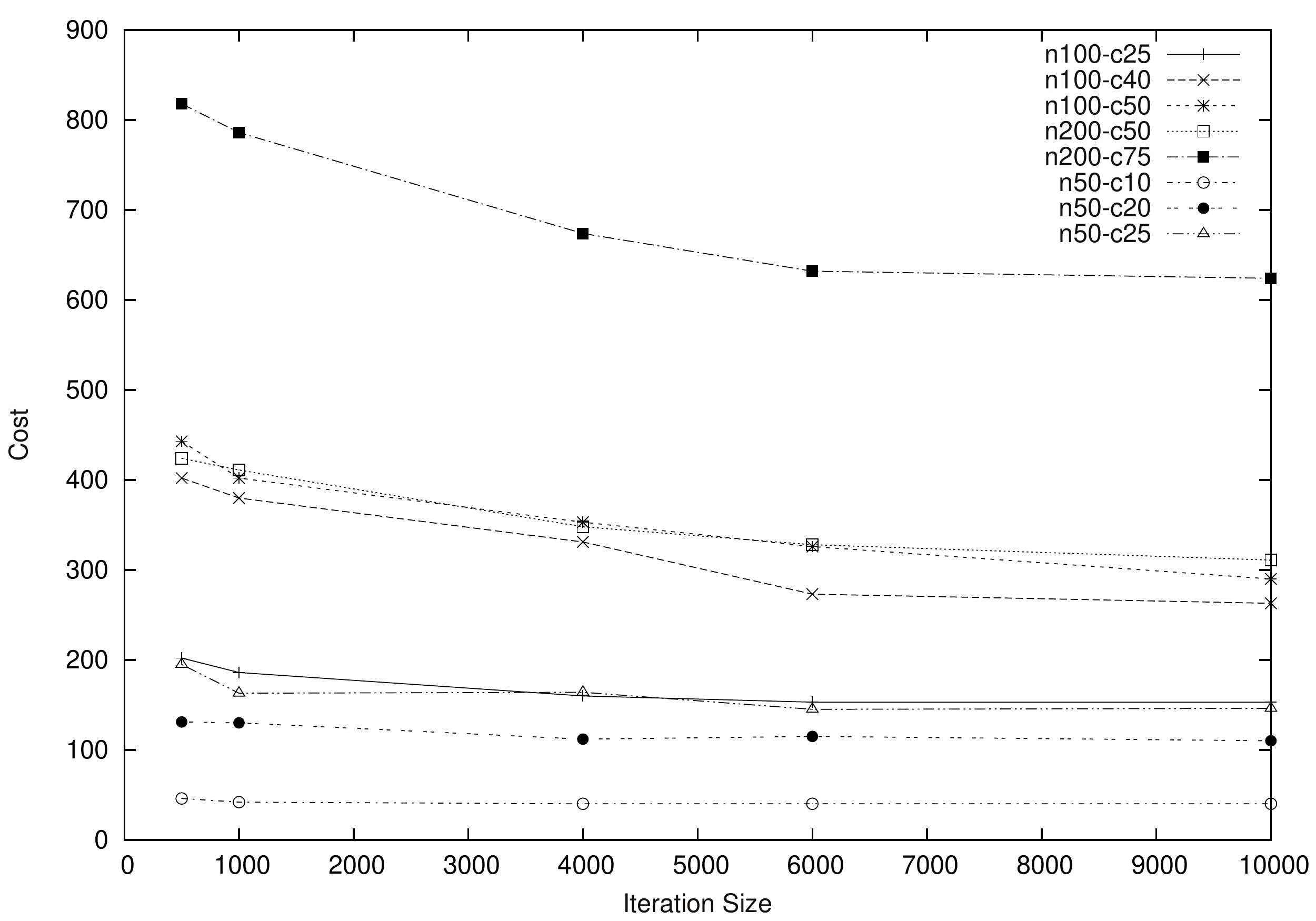}
\caption{Minimum cost for various iteration sizes in \emph{GA-ACSP}.}
\label{itersize-cost-min}
\end{figure}

\begin{figure} [h!]
\centering
\includegraphics[width=\columnwidth] {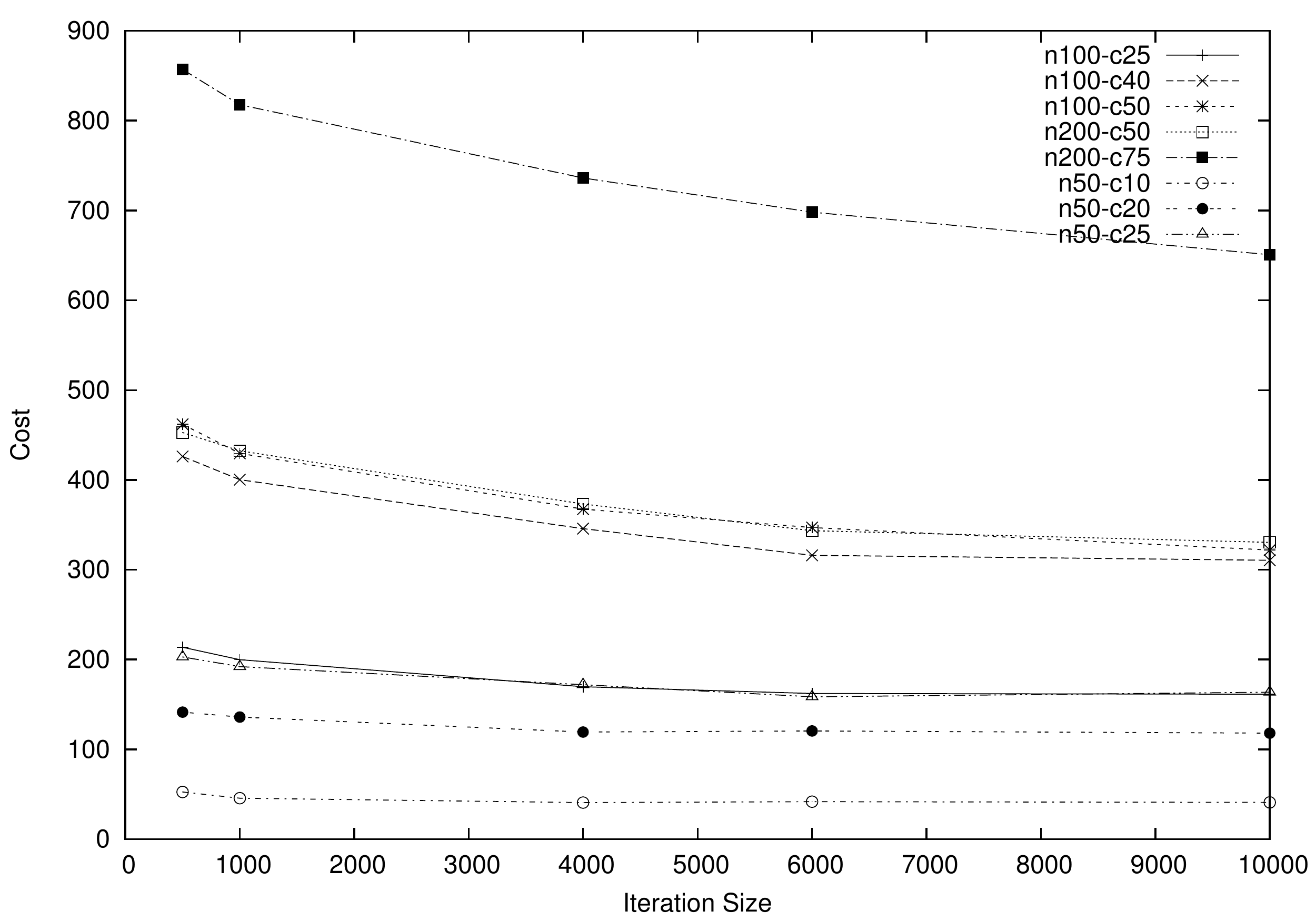}
\caption{Average cost for various iteration sizes in \emph{GA-ACSP}.}
\label{itersize-cost-avg}
\end{figure}

\begin{figure} [h!]
\centering
\includegraphics[width=\columnwidth] {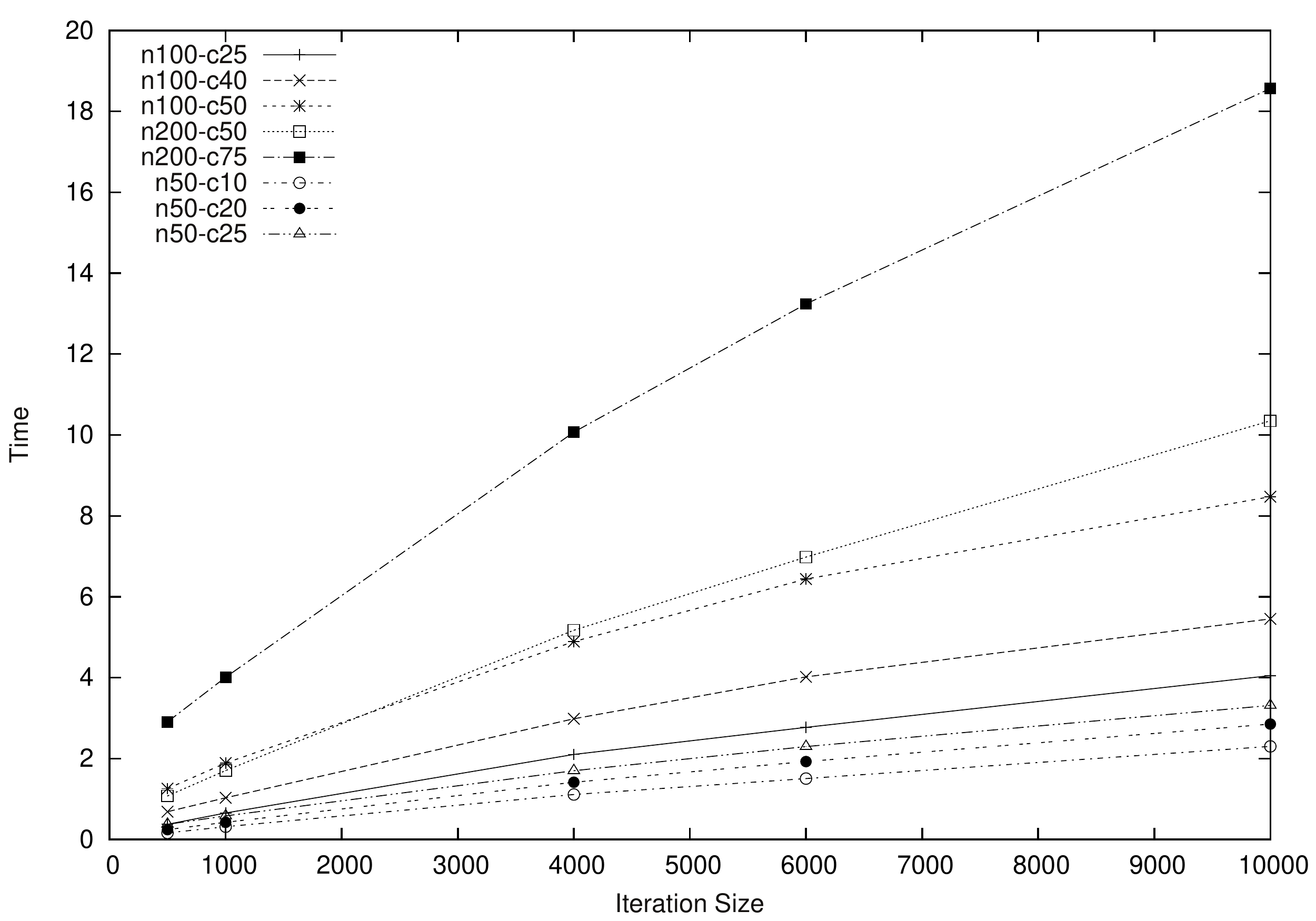}
\caption{CPU time for calculating the average cost for various iteration sizes
in \emph{GA-ACSP}.}
\label{itersize-time-avg}
\end{figure}

We also experimented on various population sizes to find the best population
size for \emph{GA-ACSP}. The results of the experiments are
presented in Figures \ref{popsize-cost-min}, \ref{popsize-cost-avg}, and
\ref{popsize-time-avg}.
Based on the results, the population size is selected as
$600$ as both the running time, and the cost at this value of population size
are lower than they are at larger population sizes.

\begin{figure} [h!]
\centering
\includegraphics[width=\columnwidth] {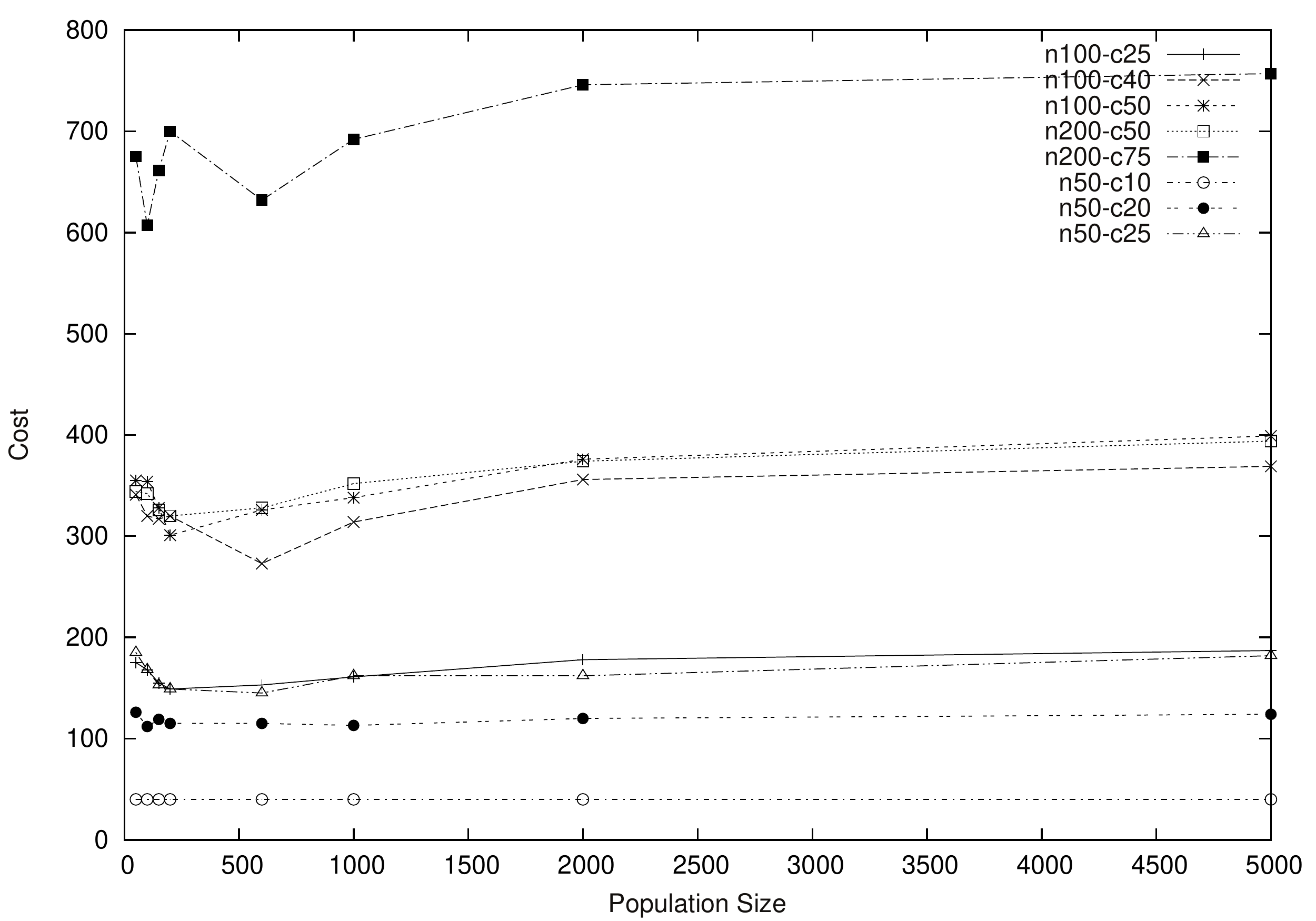}
\caption{Minimum cost for various population sizes in \emph{GA-ACSP}.}
\label{popsize-cost-min}
\end{figure}

\begin{figure} [h!]
\centering
\includegraphics[width=\columnwidth] {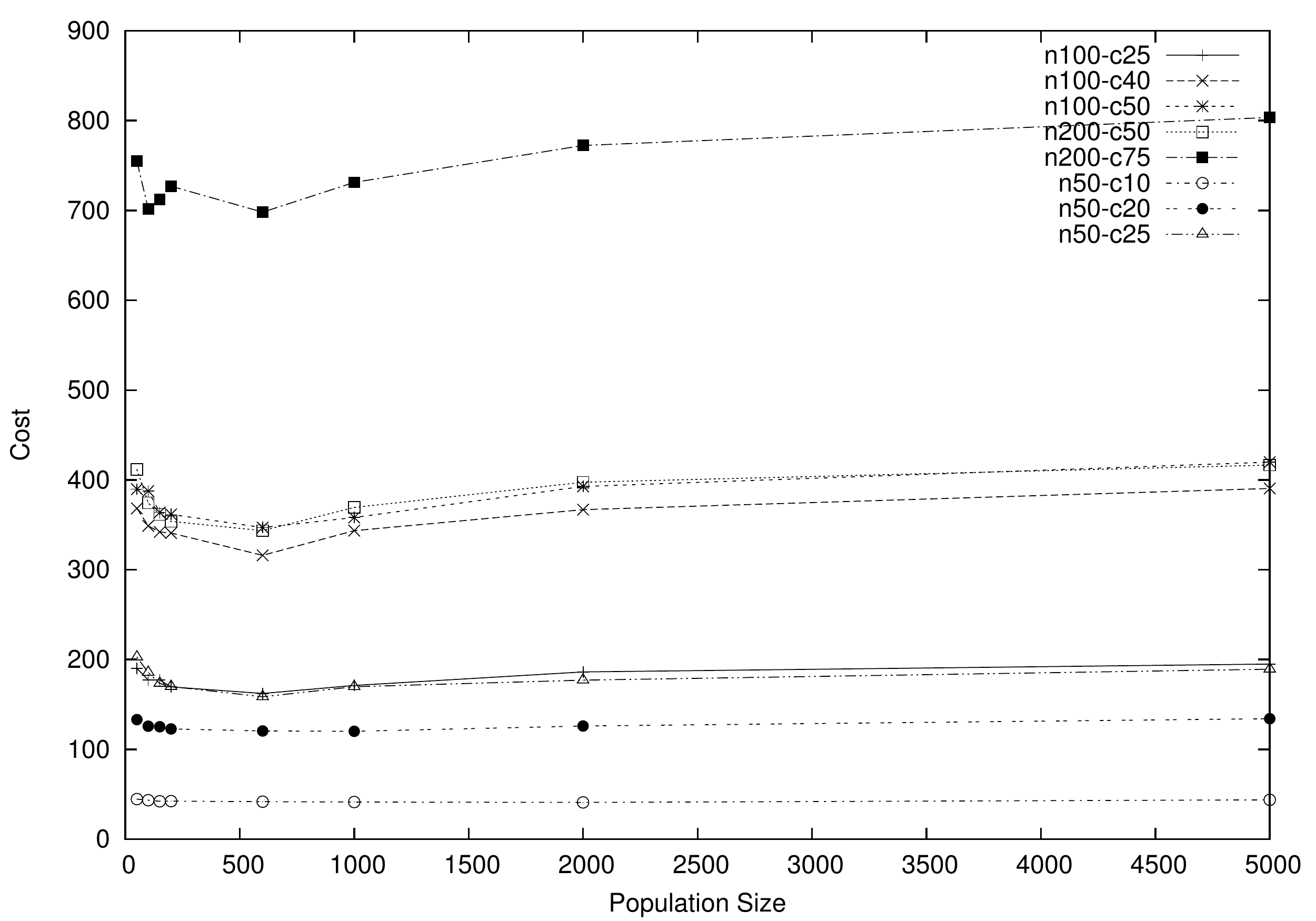}
\caption{Average cost for various population sizes in \emph{GA-ACSP}.}
\label{popsize-cost-avg}
\end{figure}

\begin{figure} [h!]
\centering
\includegraphics[width=\columnwidth] {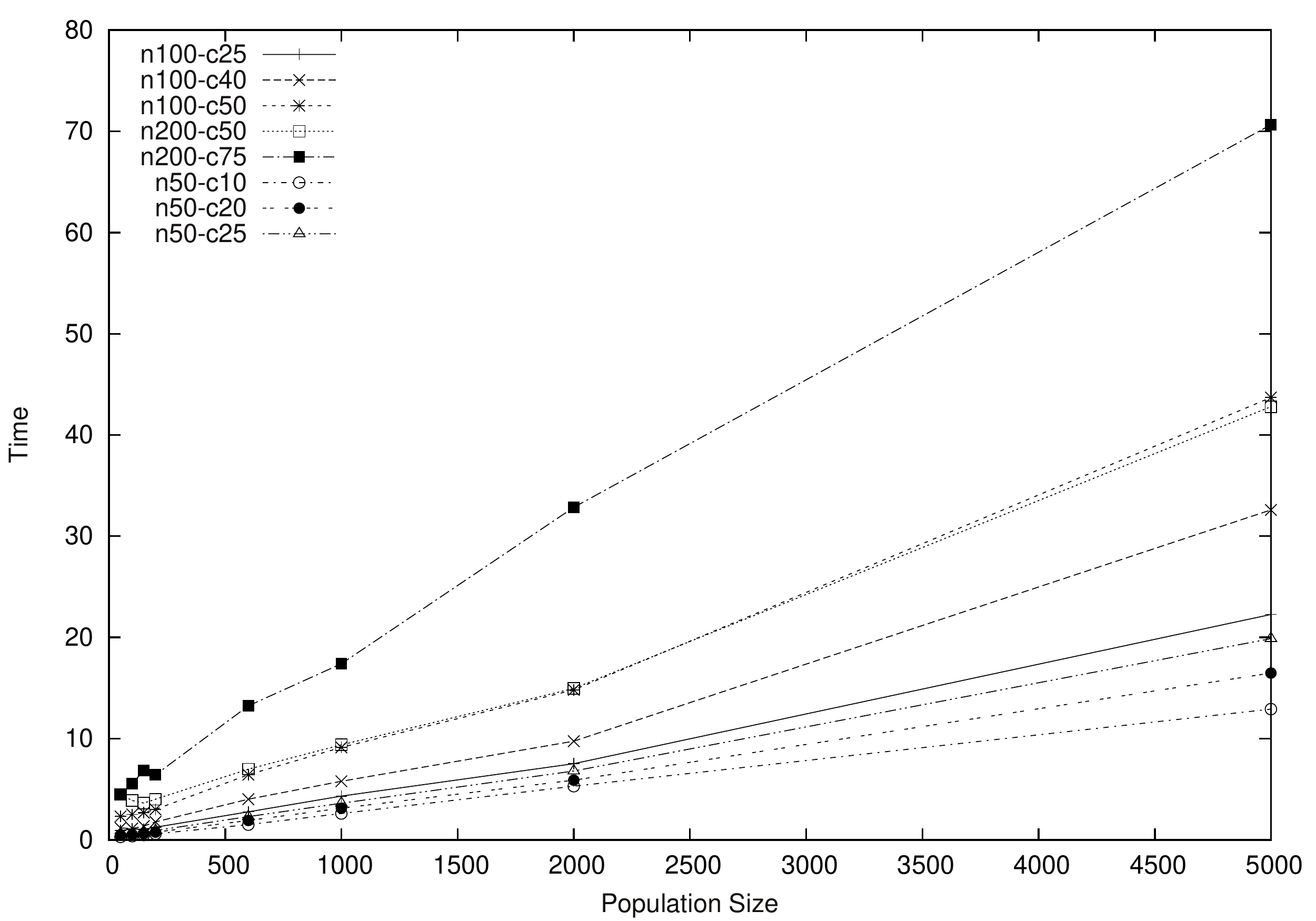}
\caption{CPU time for calculating the average cost for various population sizes
in \emph{GA-ACSP}.}
\label{popsize-time-avg}
\end{figure}

\subsection{Comparing the Heuristic Algorithms} \label{sec:comparison}

In this section, we first present the performance of the LP relaxation
based heuristics, namely
\emph{$\text{LP}_{\text{x}}\text{ACSP}$},
\emph{$\text{LP}_{\text{f}}\text{ACSP}$},
and \emph{$\text{LP}_{\text{f/x}}\text{ACSP}$}
in Figure \ref{fig:ILPComparison-Table}.
The results are reported in proportion to
the optimal values obtained via the ILP formulation.
Next, we compare
the performance of the metaheuristic algorithms,
\emph{SA-ACSP},
\emph{ACO-ACSP},
and \emph{GA-ACSP}.
The results for them are presented, again in proportion to the optimal values,
in Figure~\ref{fig:Comparison-Table}.
For each individual metaheuristic algorithm, in these tests,
the best parameter values discovered are used.
We use randomly generated graphs for the types presented
in Table \ref{table:nonlin}.

\renewcommand\arraystretch{1.4}

\begin{figure*}[th!]
\centering
\begin{tabular}{ c | c | c c | c c | c c}
\hline\hline
\multicolumn{1}{c|}{}& \multicolumn{1}{c|}{\emph{}}
& \multicolumn{2}{c|}{\emph{$\text{LP}_{\text{x}}\text{ACSP}$}} & \multicolumn{2}{c|}{\emph{$\text{LP}_{\text{f}}\text{ACSP}$}} & \multicolumn{2}{c}{\emph{$\text{LP}_{\text{f/x}}\text{ACSP}$}}\\
Graph & ILP & x & Time & f & Time & $f/x$ & Time \\
Name  &     &   &      &   &      &       &      \\
\hline
n50-c10 & 40 & 3.15 & 1.763 & 1.15 & 1.03 & 2.075 & 0.967\\
n50-c20 & 101 & 1.8614 & 2.543 & 1.1584 & 1.357 & 1.1881 & 1.201\\
n50-c25 & 132 & 1.9015 & 2.761 & 1.3182 & 1.794 & 1.3636 & 1.186\\
n100-c25 & 138 & 2.1449 & 5.803 & 1.2101 & 3.136 & 1.5145 & 3.37\\
n100-c40 & 220 & 1.9955 & 8.986 & 1.3227 & 4.181 & 1.4591 & 4.056\\
n100-c50 & 233 & 1.8112 & 8.686 & 1.1631 & 4.014 & 1.1202 & 3.526\\
n200-c50 & 223 & 2.3498 & 40.872 & 1.3946 & 23.452 & 1.6906 & 23.556\\
n200-c75 & 399 & 1.9599 & 51.767 & 1.3258 & 30.716 & 1.4837 & 23.166\\
\hline
\end{tabular}
\caption{Comparison of \emph{ILP}, \emph{$\text{LP}_{\text{x}}\text{ACSP}$}, \emph{$\text{LP}_{\text{f}}\text{ACSP}$}, and \emph{$\text{LP}_{\text{f/x}}\text{ACSP}$}.}
\label{fig:ILPComparison-Table}
\end{figure*}

\begin{figure*}[th!]
\centering
\begin{tabular}{ c | c c c | c c c | c c c}
\hline\hline
\multicolumn{1}{c|}{}& \multicolumn{3}{c|}{\emph{SA-ACSP}}
& \multicolumn{3}{c|}{\emph{ACO-ACSP}} & \multicolumn{3}{c}{\emph{GA-ACSP}}\\
Graph & Average & Minimum & Average & Average & Minimum & Average & Average & Minimum & Average \\
Name & Cost & Cost & Time & Cost & Cost & Time & Cost & Cost & Time \\
\hline
n50-c10 & 1 & 1 & 33.4153 & 1.0575 & 1 & 2.32152 & 1.0725 & 1.05 & 0.422245 \\
n50-c20 & 1.1386 & 1.1089 & 47.0029 & 1.1911 & 1.1683 & 8.15708 & 1.2267 & 1.1188 & 0.546315 \\
n50-c25 & 1.1720 &  1.1212 & 55.6508 & 1.1758 & 1.1212 & 13.8652 & 1.3303 & 1.2197 & 0.716446 \\
n100-c25 & 1.2319 & 1.1812 & 203.28 & 1.2572 & 1.1739 & 14.0686 & 1.2167 & 1.1521 & 0.957005 \\
n100-c40 & 1.3264 & 1.2182 & 260.757 & 1.3345 & 1.2773 & 38.0068 & 1.5386 & 1.4045 & 1.30289 \\
n100-c50 & 1.3914 & 1.3176 & 334.878 & 1.4039 & 1.2918 & 69.8665 & 1.5588 & 1.4464 & 2.30258 \\
n200-c50 & 1.5143 & 1.4619 & 1052.87 & 1.4404 & 1.3901 & 63.4437 & 1.5771 & 1.4484 & 2.9976 \\
n200-c75 & 1.62314 & 1.5539 & 1490.63 & 1.5 & 1.4637 & 169.171 & 1.7170 & 1.5840 & 5.69891 \\
\hline
\end{tabular}
\caption{Comparison of \emph{SA-ACSP}, \emph{ACO-ACSP}, and \emph{GA-ACSP}
with their best parameters.}
\label{fig:Comparison-Table}
\end{figure*}

Based on the experimental results, it is observed that the total path length
returned by \emph{SA-ACSP} is better than that returned by \emph{ACO-ACSP}
for medium-sized graphs. In contrast, \emph{ACO-ACSP} finds lower cost paths
compared to \emph{SA-ACSP} for larger graphs. The performance of
\emph{GA-ACSP}, in terms of solution quality, is similar to the other two metaheuristics.
It has, however, a remarkable advantage in terms of time spent over the other two
on all types of graphs.

\section{Conclusion}\label{sec:conclusion}

In this paper, a novel, and generic problem, All Colors Shortest Path (\emph{ACSP})
problem, has been formulated, and computationally explored. \emph{ACSP}
has been shown to be NP-hard, and also inapproximable within a constant factor
of the optimal. An ILP formulation has been developed for \emph{ACSP}.
Various heuristic solutions have then been devised, based on iterative rounding
applied to an LP relaxation of the ILP formulation. Moreover, three different metaheuristic
solutions based on simulated annealing, ant colony optimization, and genetic algorithm have
been proposed. Through extensive simulations, an experimental evaluation of all the heuristics
have also been reported.

The study of the computational characteristics of \emph{ACSP} when the
underlying graph is restricted to be a tree is a future work. Investigation
of an approximation bound is left as an interesting open problem.

\section{Acknowledgement}
We would like to thank Mehmet Berkehan Ak\c{c}ay for the experiments regarding
ILP.

\bibliographystyle{abbrv}
\bibliography{BASE}

\begin{thebibliography}{10}

\bibitem{BM02}
A.~Behzad and M.~Modarres.
\newblock New efficient transformation of the generalized traveling salesman
  problem into traveling salesman problem.
\newblock In {\em 15th International Conference of Systems Engineering, Las
  Vegas, USA}, 2002.

\bibitem{Dorigo1997}
M.~Dorigo and L.~M. Gambardella.
\newblock Ant colony system: A cooperative learning approach to the traveling
  salesman problem.
\newblock {\em IEEE Transactions on Evolutionary Computation}, 1(1):53--66,
  1997.

\bibitem{Dorigo1996}
M.~Dorigo, V.~Maniezzo, and A.~Colorni.
\newblock {Ant System: Optimization by a Colony of Cooperating Agents}.
\newblock {\em IEEE Transactions on Systems, Man, and Cybernetics Part
  B:Cybernetics}, 26(1):29--41, 1996.

\bibitem{DHC00}
M.~Dror, M.~Haouari, and J.~Chaouachi.
\newblock Generalized spanning trees.
\newblock {\em European Journal of Operational Research}, 120(3):583--592,
  2000.

\bibitem{FLL02}
C.~Feremans, M.~Labbe, and G.~Laporte.
\newblock A comparative analysis of several formulations for the generalized
  minimum spanning tree problem.
\newblock {\em Networks}, 39(1):29--34, 2002.

\bibitem{FGT97}
M.~Fischetti, J.~J.~S. Gonzalez, and P.~Toth.
\newblock A branch-and-cut algorithm for the symmetric generalized traveling
  salesman problem.
\newblock {\em Operations Research}, 45(3):378--394, 1997.

\bibitem{GJ79}
M.~R. Garey and D.~S. Johnson.
\newblock {\em Computers and Intractability: A Guide to the Theory of
  NP-Completeness}.
\newblock W. H. Freeman, 1979.

\bibitem{Goldberg}
D.~Goldberg.
\newblock {\em Genetic Algorithms in Search, Optimization and Machine
  Learning}.
\newblock Addison-Wesley, 1989.

\bibitem{Holland}
J.~H. Holland.
\newblock {\em Adaptation in Natural and Artificial Systems}.
\newblock The University of Michigan Press, Ann Arbor, 1975.

\bibitem{IRW99}
E.~Ihler, G.~Reich, and P.~Widmayer.
\newblock Class steiner trees and vlsi-design.
\newblock {\em Discrete Applied Mathematics}, 90(1–3):173--194, 1999.

\bibitem{Kirkpatrick1983}
S.~Kirkpatrick, C.~D. Gelatt, and M.~P. Vecchi.
\newblock Optimization by simulated annealing.
\newblock {\em Science}, 220(4598):671--680, 1983.

\bibitem{L69}
H.~Labordere.
\newblock The record balancing problem: A dynamic programming solution of a
  generalized travelling salesman problem.
\newblock {\em RAIRO Operations Research B2}, pages 43--49, 1969.

\bibitem{LMN87}
G.~Laporte, H.~Mercure, and Y.~Nobert.
\newblock Generalized travelling salesman problem through n sets of nodes: the
  asymmetrical case.
\newblock {\em Discrete Applied Mathematics}, 18(2):185--197, 1987.

\bibitem{LMW93}
Y.-N. Lien, E.~Ma, and B.~W.-S. Wah.
\newblock Transformation of the generalized traveling-salesman problem into the
  standard traveling-salesman problem.
\newblock {\em Information Sciences}, 74(1-2):177--189, 1993.

\bibitem{Metropolis}
N.~Metropolis, A.~W. Rosenbluth, M.~Rosenbluth, A.~H. Teller, and E.~Teller.
\newblock Equation of state calculations by fast computing machines.
\newblock {\em J. Chem. Phys.}, 21:1087--1092, 1953.

\bibitem{MLT95}
Y.-S. Myung, C.-H. Lee, and D.-W. Tcha.
\newblock On the generalized minimum spanning tree problem.
\newblock {\em Networks}, 26(4):231--241, 1995.

\bibitem{P04}
P.~C. Pop.
\newblock New models of the generalized minimum spanning tree problem.
\newblock {\em Journal of Mathematical Modelling and Algorithms},
  3(2):153--166, 2004.

\bibitem{PKS06}
P.~C. Pop, W.~Kern, and G.~Still.
\newblock A new relaxation method for the generalized minimum spanning tree
  problem.
\newblock {\em European Journal of Operational Research}, 170(3):900--908,
  2006.

\bibitem{PSK05}
P.~C. Pop, G.~Still, and W.~Kern.
\newblock An approximation algorithm for the generalized minimum spanning tree
  problem with bounded cluster size.
\newblock In H.~Broersma, M.~Johnson, and S.~Szeider, editors, {\em ACiD},
  volume~4 of {\em Texts in Algorithmics}, pages 115--121. King's College,
  London, 2005.

\end{thebibliography}
\end{document}